\def\year{2020}\relax
\documentclass[letterpaper]{article} 
\usepackage{aaai20}  
\usepackage{times}  
\usepackage{helvet} 
\usepackage{courier}  
\usepackage[hyphens]{url}  
\usepackage{graphicx} 
\usepackage{graphicx}  

\usepackage{booktabs}
\usepackage{algorithmic}
\usepackage{amsmath,amssymb,amsfonts}
\usepackage{multirow}
\usepackage{enumitem}
\usepackage{makecell}
\usepackage{subfigure}

\newtheorem{myPara}{Paradox}
\newtheorem{myDef}{Definition}
\newtheorem{myExp}{Example}
\newtheorem{thm}{Theorem}

\newtheorem{proof}{Proof}

\title{Sentiment Paradoxes in Social Networks: \\Why Your Friends Are More Positive Than You? }
\author{Xinyi Zhou, Shengmin Jin, Reza Zafarani\\
Data Lab, Department of EECS, Syracuse University\\
\{zhouxinyi, shengmin, reza\}@data.syr.edu
}
 
\begin{document}

\maketitle

\begin{abstract}
Most people consider their friends to be more positive than themselves, exhibiting a \textbf{ \textit{Sentiment Paradox}}. Psychology research attributes this paradox to human cognition bias. With the goal to understand this phenomenon, we study sentiment paradoxes in social networks. Our work shows that social connections (friends, followees, or followers) of users are indeed (not just illusively) more positive than the users themselves. This is mostly due to positive users having more friends. We identify five sentiment paradoxes at different network levels ranging from triads to large-scale communities. Empirical and theoretical evidence are provided to validate the existence of such sentiment paradoxes. By investigating the relationships between the sentiment paradox and other well-developed network paradoxes, i.e., \textit{friendship paradox} and \textit{activity paradox}, we find that user sentiments are positively correlated to their number of friends but rarely to their social activity. Finally, we demonstrate how sentiment paradoxes can be used to predict user sentiments.
\end{abstract}

\section{Introduction} 
\label{sec::intro}
\emph{Sentiment analysis}, also known as \emph{opinion mining}, analyzes individual opinions, sentiments, and attitudes towards various entities such as individuals, products, organizations, and topics~\cite{liu2012sentiment}. Relying on advancements in natural language processing and machine learning~\cite{ravi2015survey}, existing studies in sentiment analysis have made substantial progress towards classifying and predicting sentiments of \underline{independent} individuals and groups in social networks, focusing on tasks such as content sentiment prediction and review spam detection~\cite{breck2017opinion}. 

However, existing studies have less explored sentiments among \emph{interacting} users as their sentiments may be dependent. 
With the unavoidable peer influence in social networks~\cite{lewis2012social}, it is essential to consider user interactions when studying their sentiments, especially in large-scale social networks. 
For example, Lin et al. find that the stress levels of users are closely related to that of their friends on social media~\cite{lin2017detecting}. 
A common observation with respect to sentiments of interacting users is that many users feel their friends are more positive than themselves, experiencing a \textit{sentiment paradox}. 
There have been many discussions on why this phenomenon takes place, with psychology research linking it to human cognition biases. For example, Jordan et al.~\shortcite{jordan2011misery} show that most people have a tendency to underestimate the negative feelings of others. With many users in social networks experiencing a sentiment paradox -- being less positive than their friends -- can we attribute all such perceptions to human cognition biases alone? In other words, \textit{do sentiment paradoxes exist not only in user cognition, but also in reality?} \vspace{1mm}

\noindent\textbf{The Present Work: Sentiment Paradoxes in Networks.} We investigate whether users are indeed less positive than their social connections (friends, followees, or followers) in social networks. Possible interpretations for the existence (or non-existence) of the sentiment paradox are provided by mining the relationships between sentiment paradoxes and other well-established network paradoxes, i.e., friendship paradox and activity paradox. Finally, as an application, we show how sentiment paradoxes can be used to predict user sentiments (positive or negative).

Overall, the specific contributions of this paper are:
\begin{enumerate}
\item Five sentiment paradoxes are identified in both undirected (friend) and directed (follower and followee) social networks and at multiple network levels (triad-, community-, and network-level). Our work shows that for most users their friends are indeed more positive than them, mostly due to the fact that more positive users are more likely to have more friends, followers, and followees; 
\item We empirically and mathematically verify each paradox, where our mathematical analysis allows us to determine whether such a paradox is \textit{expected} to exist;
\item We investigate the connections between the sentiment paradox and two other well-established network paradoxes: friendship paradox and activity paradox. Our results reveal factors that can determine the (1) existence and the (2) magnitude of sentiment paradoxes; and
\item We demonstrate the role that the sentiment paradox that can play in practical applications, i.e., in predicting s user's sentiments by looking at the sentiments of his or her social connections.
\end{enumerate}

The remainder of the work is organized as follows. Experimental setup is presented first, followed by a formal definition for the general sentiment paradox, and sentiment paradoxes in triads and communities. Then, we investigate the connections of sentiment paradox to other network paradoxes, which helps determine the existence and magnitude of sentiment paradoxes. One application of sentiment paradoxes, i.e., user sentiment prediction, is provided next. Finally, a literature review and some conclusions are provided.

\section{Experimental Setup}
\label{sec:experiments}

To study sentiment paradoxes at different network levels, proper data that contains user sentiments and their network information (e.g., friends or communities joined) is required.

\vspace{1em}
\noindent
\textbf{Dataset.} 
\footnote{The data is released at: \url{http://data.syr.edu/get/EmotionPatterns}}
We have crawled a large-scale dataset from LiveJournal~\cite{zafarani2009social,jin2017emotions}. LiveJournal is a popular blogging and social networking site, where users can maintain a blog, journal, or a diary. Data collected from LiveJournal has several advantages:
\begin{enumerate}
\item Sentiments are directly provided: when posting blogs, users can report their sentiments by selecting a \emph{mood} (e.g., excited, busy or angry, see Appendix for a sample user post with its mood), which provides access to sentiment ground truth; 
\item Both undirected (friends) and directed (followees and followers) relationships of users exist, i.e., a directed and an undirected network. Note that these relationships are separate: a user can choose to subscribe (follow) another person without approval, and/or befriend (with approval) so that the two users can share some private posts. Hence, two users can follow each other (i.e., two directed edges in the directed network), but not be friends (no edge between them in the undirected network); and 
\item Community membership information is explicitly available on user profiles (i.e., no need to detect them using community detection, which can be subjective~\cite{fortunato2010community}). User can decide to create or join communities. Each community is often related to some topic and users in the same community often share similar interests. 
\end{enumerate}

We have collected the following data spanning more than 10 years (from 1999 to 2010)~\cite{zafarani2009social,jin2017emotions}:
(i)~users and their posts to obtain user sentiments;
(ii)~friendships and followee/follower relationships among users; and
(iii)~community memberships for all users. 
We only retain users with ten or more posts to exclude occasionally active or inactive users. We plot the post distribution of these excluded users, which is provided in the Appendix and indicates that most ($\sim$97\%) of these users have posted nothing. As moods are limited (132 moods), we manually convert each mood in our dataset to its sentiment polarity (details can be seen in the Appendix): positive~($+$, e.g., cheerful, excited and happy), negative~($-$, e.g., angry, annoyed and depressed) or neutral~($0$, e.g., busy). 
Some statistics on our data is provided in Table~\ref{tab::statistics}. 

\begin{table}[t]
\centering
\caption{Data Statistics}
\label{tab::statistics}
\begin{tabular}{lr}
 \toprule[1pt]
 \textbf{Data} & \textbf{Number} \\ \hline
 {\# Users} 		& 115,444         \\ 
 {\# User Posts}  & 12,404,868      \\ 
 {\# Friendships}							& 246,164         \\ 
 {\# Followee/Follower Relationships} & 793,948         \\ 
 {\# Triads (Undirected)} 				& 262,036         \\ 
 {\# Triads (Directed)}					& 7,264,770       \\ 
 {\# Communities} 	& 200,208         \\
 {\# Community Memberships}				& 2,473,074       \\ \bottomrule[1pt]
\end{tabular}
\end{table}

\vspace{1em}
\noindent
\textbf{User Sentiments.} Traditionally, to obtain user sentiment values, one can rely on self-assessment surveys, which is time-consuming for large number of users. Here, we adopt an automatic way by investigating the historical posts of users (see Definition~\ref{def::swb})~\cite{bollen2011happiness}. 

\begin{myDef}[Subjective Well-Being (SWB)\footnotemark]
\label{def::swb}
Assume user $u$ has $N_p(u)$ positive posts and $N_n(u)$ negative posts. The SWB value of $u$, denoted by $S(u)$, is
\begin{equation}
S(u) = \frac{N_p(u)-N_n(u)}{N_p(u)+N_n(u)}.
\end{equation}
\end{myDef}

Note that $S(u)\in[-1,+1]$, where $-1$ shows an extremely negative user and $+1$, an extremely positive user.

\footnotetext{Strictly, what we study is a component of the SWB rather than SWB itself as it includes both affective and cognitive parts.}

\vspace{1em}
\noindent
\textbf{Sentiment Distribution.} The distribution of user sentiments can be obtained by plotting the SWB distribution. As observed in Figure~\ref{fig::swbDist}, the SWB distribution approximately follows a normal distribution $\mathcal{N}(\mu,\sigma^2)$, which aligns with findings on sentiment distributions in other social networks (e.g., that of Twitter~\cite{ferrara2015quantifying}).
Using a normal fit, we obtain the SWB distribution, which is $\mathcal N(0,0.08)$.

\begin{figure}[t]
    \centering
    \includegraphics[width=\columnwidth]{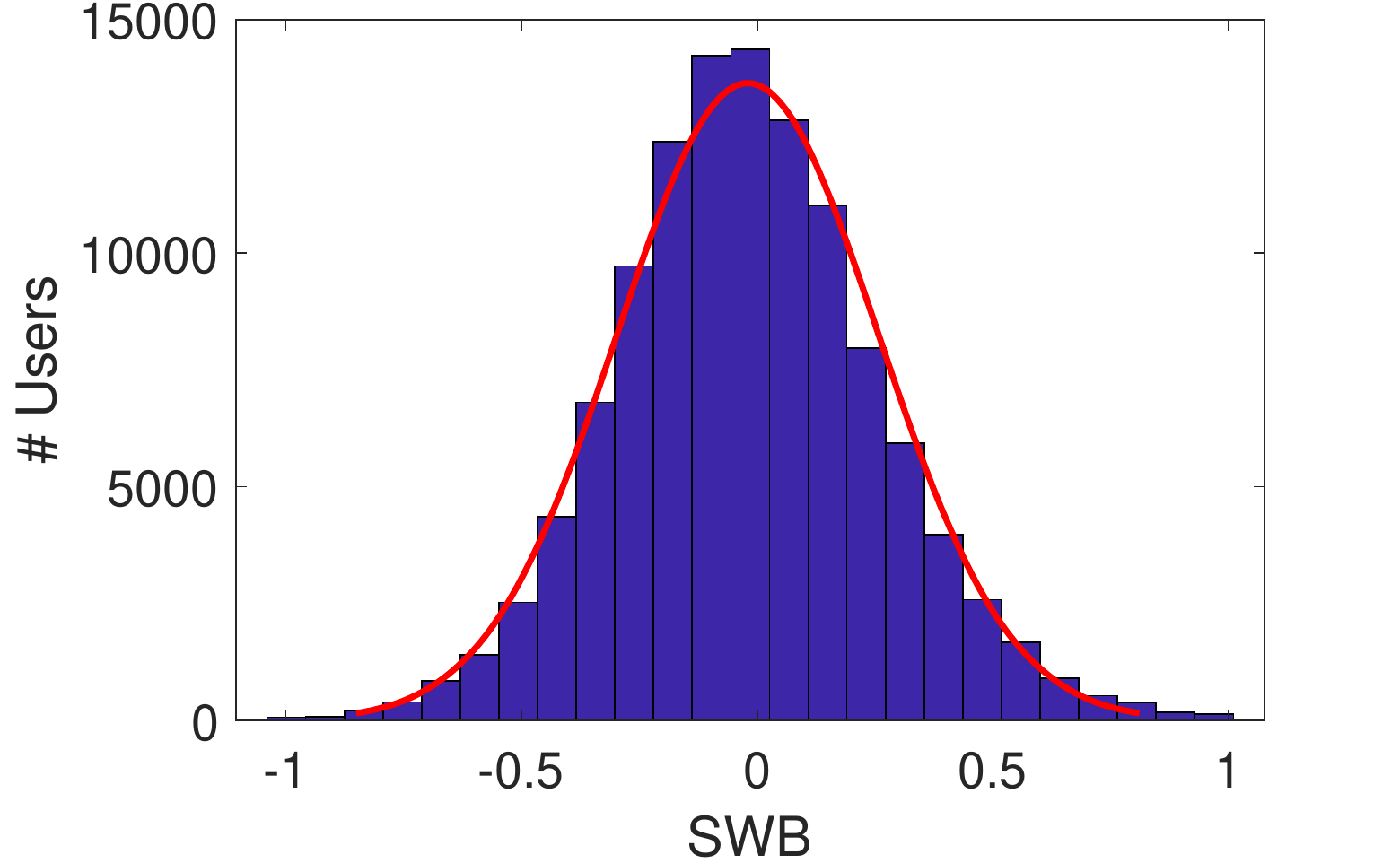}
    \caption{User Sentiment Distribution (SWB values)}
    \label{fig::swbDist}
\end{figure}

\section{Sentiment Paradox}
\label{sec::gsp}
In this section, we mainly focus on a ``general'' sentiment paradox, which can be observed among all users of a network. We first present the definition of sentiment paradox, followed by experiments to verify its existence and mathematical proofs on whether the paradox is \textit{expected} to exist. 

\vspace{1em}
\noindent
\textbf{Definition.} The sentiment paradox, or \textit{network sentiment paradox}, can be summarized as

\begin{myPara}[Sentiment Paradox]
\label{par::gsp}
Your friends, followees, or followers are more positive than you.
\end{myPara}

\vspace{1em}
\noindent
\textbf{Empirical Verification.} 
To verify whether the sentiment paradox exists, we take the following three steps:\vspace{1mm}

\noindent \textit{I. User Sentiment Assignment.} We calculate how positive or negative users are by computing their SWB (Definition \ref{def::swb}).\vspace{1mm}

\noindent \textit{II. Computing Paradox Magnitude.} Consider a user whose SWB value is less than the (i) mean or (ii) median of the SWB values of his or her connections. We can consider three types of connections: friends, followees, or followers. We consider this user as being less positive than his or her connections and denote the proportion of such users in a social network as the sentiment paradox \emph{magnitude}:

\begin{myDef}[Sentiment Paradox Magnitude]
\label{def::paraMagnitude}
Consider a social network with a set of users $U=\{u_i\}$, $i=1,2,\cdots,n$, each with a SWB value $S(u_i)$. For each user $u_i$, we denote her connections (either friends, followers, or followees) by $c_{ij}$, $j=1,2,\cdots,m$. 
The sentiment paradox magnitude of the network is calculated by
\begin{equation}
M=\frac{\sum_{u_i}I(\,S(u_i) < \bar{S}(c_{ij})\,)}{||U||},
\end{equation}
where $I(a<b)=1$ when $a<b$ and is 0 otherwise. The value $\bar{S}(c_{ij})$ is the (i) mean or the (ii) median of $S(c_{ij})$'s.
\end{myDef} 

When the magnitude is greater than 0.5, we say the sentiment paradox \emph{strongly holds} in the network. When the magnitude is less than or equal to 0.5, but is still greater than the proportion of users that are more positive than their connections, and that of users that are as positive as their connections, we say the paradox \emph{weakly holds} in the network.\vspace{1mm}

\noindent \textit{III. Assessing Statistical Significance.} To assess the statistical significance of our findings, we compute the difference between the \textit{observed} and \textit{expected} paradox magnitudes. To compute the expected paradox magnitude, we maintain the SWB distribution of users and their network structure, but randomly assign a SWB value to each user. After random assignments, we recalculate the paradox magnitude. We conduct this experiment 1,000 times, and record the average magnitude, which is the \emph{expected} paradox magnitude.
To assess how significant the difference between observed and expected paradox magnitudes is, we compute \emph{surprise}~\cite{leskovec2010signed}:

\begin{myDef}[Surprise]
\label{def::surprise}
In a social network with $N$ users, if paradox magnitude is $M$ and expected paradox magnitude is $M_{\mathsf{Expected}}$ ($M_{\mathsf{Expected}}\neq 0$ and $1$), the surprise value is
\begin{equation}
\mathsf{surprise} = \frac{N(M-M_{\mathsf{Expected}})}{\sqrt{N\cdot M_{\mathsf{Expected}} (1-M_{\mathsf{Expected}})}}.
\end{equation}
\end{myDef}

 A surprise value on the order of tens is highly significant, indicating that $p$-values are nearly zero.\vspace{2mm}

Following this three-step process, we obtain the results in Table~\ref{tab::resultGSP}, where ``Holds'' (``Does not hold'') indicates that users are less (more) positive than their connections. ``Unknown'' indicates that users are as positive as their connections. 
In both undirected and directed networks, irrespective of using mean or the median, we make the following three observations:
\begin{enumerate}
    \item Sentiment paradox strongly holds within the network, as the observed sentiment paradox magnitudes (user proportions) for all networks are greater than 0.5; 
    \item The observed paradox magnitude values are all higher than the expected paradox magnitudes; and 
    \item The surprise values are on the order of tens, which indicate that the observed paradox magnitudes are all statistically significant.
\end{enumerate}  

\begin{table}[t]
\centering
\caption{Empirical Verification of Sentiment Paradox. The observed proportions are greater than 0.5 indicates that the sentiment paradox strongly holds within networks. The observed proportions are higher than the expected ones, where such difference is statistically significant as the surprise values are on the order of tens.}
\label{tab::resultGSP}
\subtable[User Sentiments vs. \underline{Average} Sentiments of Connections]{
\centering
\label{subtab::meanGSP}
\resizebox{.98\columnwidth}{!}{
\begin{tabular}{clrrrr}
 \toprule[1pt]
\multicolumn{2}{c}{\multirow{2}{*}{{\textbf{Sentiment Paradox}}}} & \multicolumn{2}{c}{{\textbf{Observed}}} & \multicolumn{1}{c}{\textbf{Exp.}} & \multirow{2}{*}{{\textbf{Surprise}}} \\ \cline{3-4} 
\multicolumn{2}{c}{} & \multicolumn{1}{c}{{\textbf{\#Users}}} & \multicolumn{1}{c}{{\textbf{Prop.}}} & \multicolumn{1}{c}{\textbf{Prop.}} &  \\ \hline 
\multirow{4}{*}{\rotatebox{90}{\small\textbf{Friends}}} & {\textbf{Holds}} & 43,786 & 55.11\% & 50.10\%   & 28.21 \\  
 & {\textbf{Does not hold}} & 35,588 & 44.79\% & 49.90\% & -28.77 \\  
 & {\textbf{Unknown}} & 79 & 0.10\% & 0.00\%   & - \\  \cline{2-6}
 & {\textbf{Total}} & 79,453 & 100\% & 100\%   & - \\ \midrule[0.5pt] 
\multirow{4}{*}{\rotatebox{90}{\small\textbf{Followees}}} & {\textbf{Holds}} & 44,336 & 54.11\% & 49.50\%   & 26.41 \\ 
 & {\textbf{Does not hold}} & 36,699 & 44.79\% & 49.41\%  & -26.50 \\
 & {\textbf{Unknown}} & 906 & 1.10\% & 1.09\%  & 0.47 \\ \cline{2-6}
 & {\textbf{Total}} & 81,941 & 100\% & 100\%   & -  \\ \midrule[0.5pt]  
\multirow{4}{*}{\rotatebox{90}{\small \textbf{Followers}}} & {\textbf{Holds}} & 26,287 & 52.87\% & 49.58\%  & 14.67 \\  
 & {\textbf{Does not hold}} & 23,015 & 46.29\% & 49.60\% & -14.74 \\
 & {\textbf{Unknown}} & 420 & 0.84\% & 0.82\%  & 0.40 \\ \cline{2-6} 
 & {\textbf{Total}} & 49,722 & 100\% & 100\% & - \\ \bottomrule[1pt]
\end{tabular}}} 
\subtable[User Sentiments vs. \underline{Median} Sentiments of Connections]{
\centering
\label{subtab::medianGSP}
\resizebox{.98\columnwidth}{!}{
\begin{tabular}{clrrrr}
\toprule[1pt]
\multicolumn{2}{c}{\multirow{2}{*}{\textbf{Sentiment Paradox}}} & \multicolumn{2}{c}{{\textbf{Observed}}} & \multicolumn{1}{c}{\textbf{Exp.}} & \multirow{2}{*}{{\textbf{Surprise}}} \\ \cline{3-4} 
\multicolumn{2}{c}{} & \multicolumn{1}{c}{{\textbf{\#Users}}} & \multicolumn{1}{c}{{\textbf{Prop.}}} & \multicolumn{1}{c}{\textbf{Prop.}} &  \\ \hline 
\multirow{4}{*}{\rotatebox{90}{\small\textbf{Friends}}} & {\textbf{Holds}} & 43,621 & 54.90\% & 50.00\%  & 27.61 \\  
 & {\textbf{Does not hold}} & 35,684 & 44.91\% & 49.96\%  & -28.44 \\ 
 & {\textbf{Unknown}} & 148 & 0.19\% & 0.04\%  & 20.76 \\  \cline{2-6}  
 & {\textbf{Total}} & 79,453 & 100\% & 100\%   & - \\ \midrule[0.5pt]
\multirow{4}{*}{\rotatebox{90}{\small\textbf{Followees}}} & {\textbf{Holds}} & 43,311 & 52.86\% & 49.00\%  & 22.13 \\  
 & {\textbf{Does not hold}} & 36,789 & 44.90\% & 48.91\% & -23.01 \\ 
 & {\textbf{Unknown}} & 1,841 & 2.24\% & 2.09\%  & 3.07 \\ \cline{2-6} 
 & {\textbf{Total}} & 81,941 & 100\% & 100\% & -  \\ \midrule[0.5pt]
\multirow{4}{*}{\rotatebox{90}{\small\textbf{Followers}}} & {\textbf{Holds}} & 25,542 & 51.37\% & 48.92\%  & 10.92 \\  
 & {\textbf{Does not hold}} & 23,073 & 46.40\% & 48.91\%  & -11.17 \\  
 & {\textbf{Unknown}} & 1,107 & 2.23\% & 2.17\%  & 0.85 \\ \cline{2-6}
 & {\textbf{Total}} & 49,722 & 100\% & 100\% & - \\ \bottomrule[1pt]
\end{tabular}}} 
\end{table}

\begin{table*}[t]
\Huge
\centering
\caption{Sentiment Paradoxes at the Triad and Community Levels}
\label{tab::triad_community}
\resizebox{\textwidth}{!}{
\begin{tabular}{clrrrr|rrrr|rrrr|rrrr}
\toprule[2pt]
 &  & \multicolumn{4}{c|}{\textbf{Triad Sentiment Paradox}} & \multicolumn{4}{c|}{\makecell{\textbf{Common-neighbor} \\ \textbf{Sentiment Paradox}}} & \multicolumn{4}{c|}{\textbf{Community Sentiment Paradox}} & \multicolumn{4}{c}{\makecell{\textbf{Common-interest} \\ \textbf{Sentiment Paradox}}} \\ \cline{3-18}
 &  & \multicolumn{2}{c}{\textbf{Observed}} & \multicolumn{1}{c}{{\textbf{Exp.}}} & \multirow{2}{*}{{\textbf{Surp.}}} & \multicolumn{2}{c}{\textbf{Observed}} & \multicolumn{1}{c}{{\textbf{Exp.}}} & \multirow{2}{*}{{\textbf{Surp.}}} & \multicolumn{2}{c}{\textbf{Observed}} & \multicolumn{1}{c}{{\textbf{Exp.}}} & \multirow{2}{*}{{\textbf{Surp.}}} & \multicolumn{2}{c}{\textbf{Observed}} & \multicolumn{1}{c}{{\textbf{Exp.}}} & \multirow{2}{*}{{\textbf{Surp.}}} \\ \cline{3-4} \cline{7-8} \cline{11-12} \cline{15-16}
 &  & \textbf{\#Users} & \multicolumn{1}{c}{{\textbf{Prop.}}} & \multicolumn{1}{c}{{\textbf{Prop.}}} &  & \textbf{\#Users} & \multicolumn{1}{c}{{\textbf{Prop.}}} & \multicolumn{1}{c}{{\textbf{Prop.}}} &  & \textbf{\#Users} & \multicolumn{1}{c}{{\textbf{Prop.}}} & \multicolumn{1}{c}{{\textbf{Prop.}}} &  & \textbf{\#Users} & \multicolumn{1}{c}{{\textbf{Prop.}}} & \multicolumn{1}{c}{{\textbf{Prop.}}} &  \\ \midrule
\multirow{4}{*}{\rotatebox{90}{\Large\textbf{Friends}}} & \textbf{Holds} & 11,044 & 52.52\% & 48.41\% & 11.91 & 11,333 & 53.89\% & 50.07\% & 11.10 & 20,381 & 51.19\% & 49.10\% & 12.21 & 20,742 & 53.12\% & 49.88\% & 12.77 \\
 & \textbf{Does not hold} & 9,298 & 44.22\% & 48.27\% & -11.74 & 9,695 & 46.11\% & 49.93\% & -11.10 & 17,887 & 45.80\% & 48.93\% & -12.37 & 18,254 & 46.75\% & 50.06\% & -13.09 \\
 & \textbf{Unknown} & 686 & 3.26\% & 3.32\% & -0.47 & 0 & 0.00\% & 0.00\% & - & 783 & 2.01\% & 1.97\% & -0.57 & 55 & 0.14\% & 0.06\% & 6.67 \\ \cline{2-18}
 & \textbf{Total} & 21,028 & 100\% & 100\% & - & 21,028 & 100\% & 100\% & - & 39,051 & 100\% & 100\% & - & 39,051 & 100\% & 100\% & - \\ \midrule
\multirow{4}{*}{\rotatebox{90}{\Large\textbf{Followees}}} & \textbf{Holds} & 37,108 & 53.58\% & 49.03\% & 23.95 & 37,698 & 54.43\% & 50.01\% & 23.27 & 19,176 & 50.44\% & 46.96\% & 13.60 & 20,081 & 52.82\% & 48.78\% & -15.75 \\
 & \textbf{Does not hold} & 30,380 & 44.51\% & 49.00\% & -23.64 & 31,564 & 45.57\% & 49.99\% & -23.27 & 16,502 & 43.40\% & 46.99\% & -14.03 & 16,985 & 44.67\% & 48.75\% & -15.89 \\
 & \textbf{Unknown} & 1,326 & 1.91\% & 1.97\% & -1.14 & 2 & 0.00\% & 0.00\% & - & 2,341 & 6.16\% & 6.05\% & 0.90 & 953 & 2.51\% & 2.47\% & 0.43 \\ \cline{2-18}
 & \textbf{Total} & 69,264 & 100\% & 100\% & - & 69,264 & 100\% & 100\% & - & 38,019 & 100\% & 100\% & - & 38,019 & 100\% & 100\% & - \\ \midrule
\multirow{4}{*}{\rotatebox{90}{\Large\textbf{Followees}}} & \textbf{Holds} & 39,496 & 52.61\% & 49.08\% & 19.35 & 39,952 & 53.22\% & 50.01\% & 17.59 & 13,314 & 49.95\% & 46.92\% & 9.91 & 13,872 & 52.04\% & 48.77\% & 10.69 \\
 & \textbf{Does not hold} & 34,183 & 45.54\% & 49.07\% & -19.35 & 35,116 & 46.78\% & 49.99\% & -17.59 & 11,733 & 44.01\% & 47.18\% & -10.37 & 12,127 & 45.49\% & 48.79\% & -10.79 \\
 & \textbf{Unknown} & 1,389 & 1.85\% & 1.85\% & 0 & 0 & 0.00\% & 0.00\% & - & 1,611 & 6.04\% & 5.90\% & 0.97 & 659 & 2.47\% & 2.44\% & 0.32 \\ \cline{2-18}
 & \textbf{Total} & 75,068 & 100\% & 100\% & - & 75,068 & 100\% & 100\% & - & 26,658 & 100\% & 100\% & - & 26,658 & 100\% & 100\% & - \\ \bottomrule[2pt]
\end{tabular}}
\end{table*}

\vspace{0.5em}
\noindent
\textbf{Theoretical Verification.} We observe empirically from Table~\ref{tab::resultGSP} that the expected magnitudes for the sentiment paradox to hold and not hold are almost the same, indicating that the paradox is not expected to exist within networks. Theorem~\ref{thm::gsp} theoretically justifies this empirical observation.

\begin{thm}
\label{thm::gsp}
If the SWB values of users in a network follow a normal distribution $\mathcal N(\mu,\sigma^2)$, the SWB value of a user is expected to be equal to the (i) mean and (ii) median of SWB values of his connections (friends, followees, or followers), i.e., a user is expected to be as positive as his connections.
\end{thm} 
\begin{proof}
Assume random variable $S \in [-1,1]$, which denotes the SWB values of users, follows a normal distribution $\mathcal N(\mu,\sigma^2)$. 
Assume we sample $n$ times from this distribution, where $n$ is the number of users in the network. For each user $u_i$,  $i=1,2,\cdots,n$, we have two sample sets: (i) $S_u^i$, with size one, as the SWB value of user $u_i$; (ii) $S_f^i$, as the SWB values of the connections (friends, followees, or followers) of user $u_i$.
Assume $\bar{S}_u$ denote the sample mean from the sample $S_u^i$, and $\bar{S}_f$ is the sample mean from samples $S_f^i$. Note that $\bar{S}_u = S(u_i)$ as $||S_u^i||=1$. $S \sim \mathcal N(\mu,\sigma^2)$ indicates $\bar{S}_u \sim \mathcal N(\mu,\sigma^2)$ and $\bar{S}_f \sim \mathcal N(\mu,\sigma_c^2)$, for some $\sigma_c$. Hence, $E(\bar{S}_u)=E(\bar{S}_f)=\mu$ and $E(\bar{S}_u-\bar{S}_f)=0$, which indicates that the expected SWB values of users are equal to the expected average SWB values of their connections. For the median, the proof is similar as the median and the mean are the same value in a normal distribution. 
\end{proof}

\section{Sentiment Paradox in Triads}
\label{sec:triadLevel}

Triads (a group of three connected people) are crucial components of networks, especially when investigating ideas such as structural balance (i.e., a friend of a friend is a friend), clusterability (i.e., friends form small groups) and transitivity (i.e., $A$ is a friend of $B$, $B$ is a friend of $C$, so $A$ is a friend $C$). In this section, we study sentiments among interacting users in triads. We will investigate if a sentiment paradox exists at the triad level, and aim to provide explanations on the existence (or lack) of such a paradox.

To explore if the sentiment paradox holds within triads, assume user $u_i$, $i=1,2,\cdots,n$ is a member of triads $t_j$, $j=1,2,\cdots,m$. Within each $t_j$, we compare the sentiment (i.e., the SWB value) of $u_i$ and the mean and median of that of his two connections (friends, followees, or followers). Note at the triad level, the results based on either the mean or the median should be the same because each user has no more than two connections within a triad.
If in the majority of triads that $u_i$ is a member of, $u_i$ is less positive than his connections, we consider $u_i$ as a user exhibiting sentiment paradox at the triad level.
Then, we compute the proportion of such users in the overall network, and conduct significance analysis similar to how it was conducted in last section. Note such paradox is not expected to exist, as proved in Theorem~\ref{thm::gsp}. However, Table~\ref{tab::triad_community} provides the empirical results, which can be summarized as:

\begin{myPara}[Triad Sentiment Paradox]
\label{par::TNSP}
{Your friends, followees, or followers in a triad are more positive than you.}
\end{myPara}

On the other hand, one can think of a triad as a pair of users sharing a \textit{common neighbor}. This observation motivates us to verify whether there is a sentiment paradox between users and their connections with whom users share common neighbors. Hence, we conduct an experiment similar to the one performed to validate the sentiment paradox in last section, except that we only compare sentiments between users and a subset of their connections with whom users share a triad, i.e., have at least one common neighbor. The paradox is not expected to exist either, as proved in Theorem~\ref{thm::gsp}. However, the empirical results in Table~\ref{tab::triad_community} show that:

\begin{myPara}[Common-neighbor Sentiment Paradox]
\label{par::triad3}
Your friends (followees or followers) with whom you share friends (followees or followers) are more positive than you.
\end{myPara}

\begin{figure*}[t]
 \centering
  \subfigure[Friendships]{
  \begin{minipage}{0.275\textwidth}
  	\centering
  	\includegraphics[width=\textwidth]{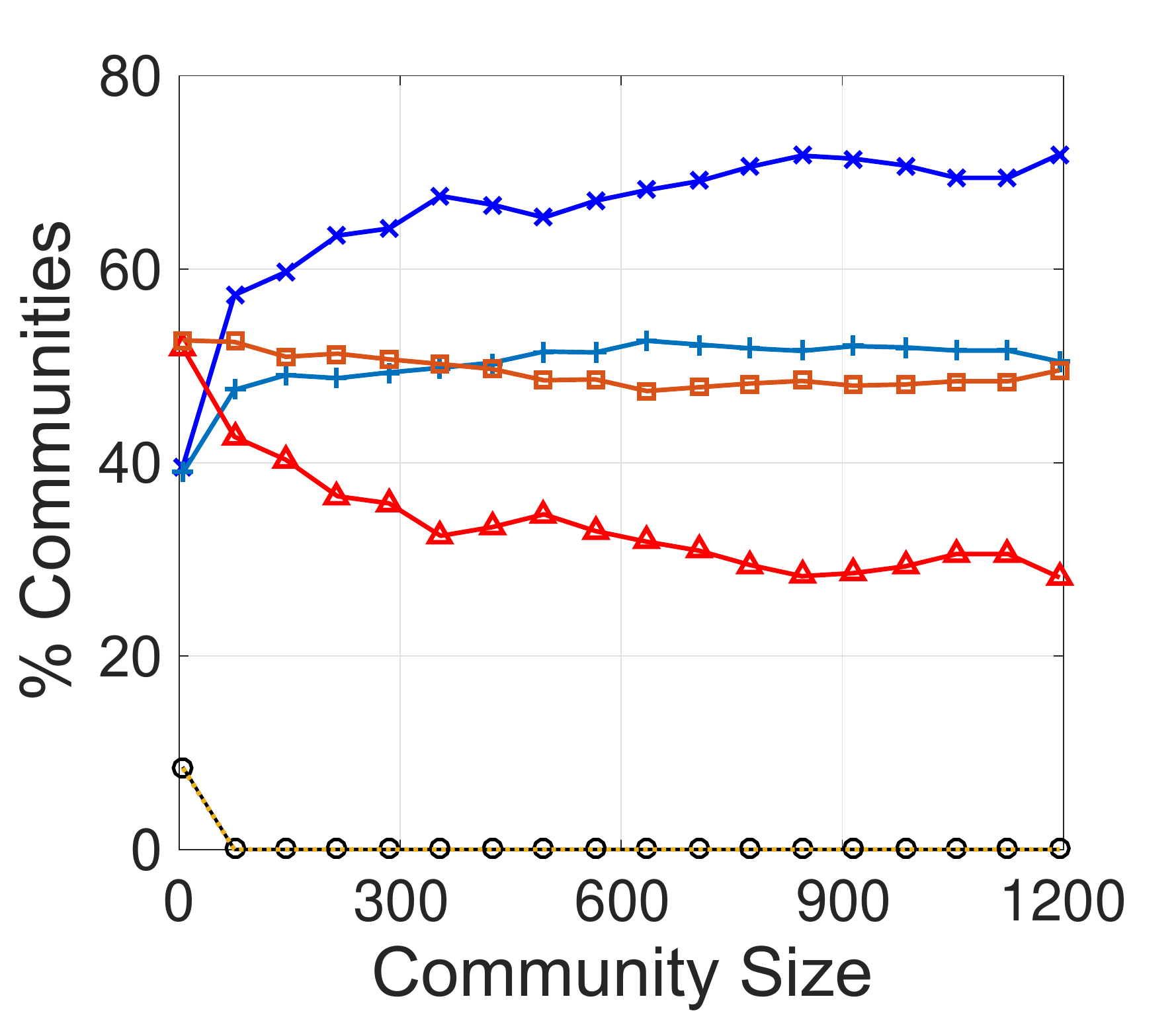}\\
  	\includegraphics[width=\textwidth]{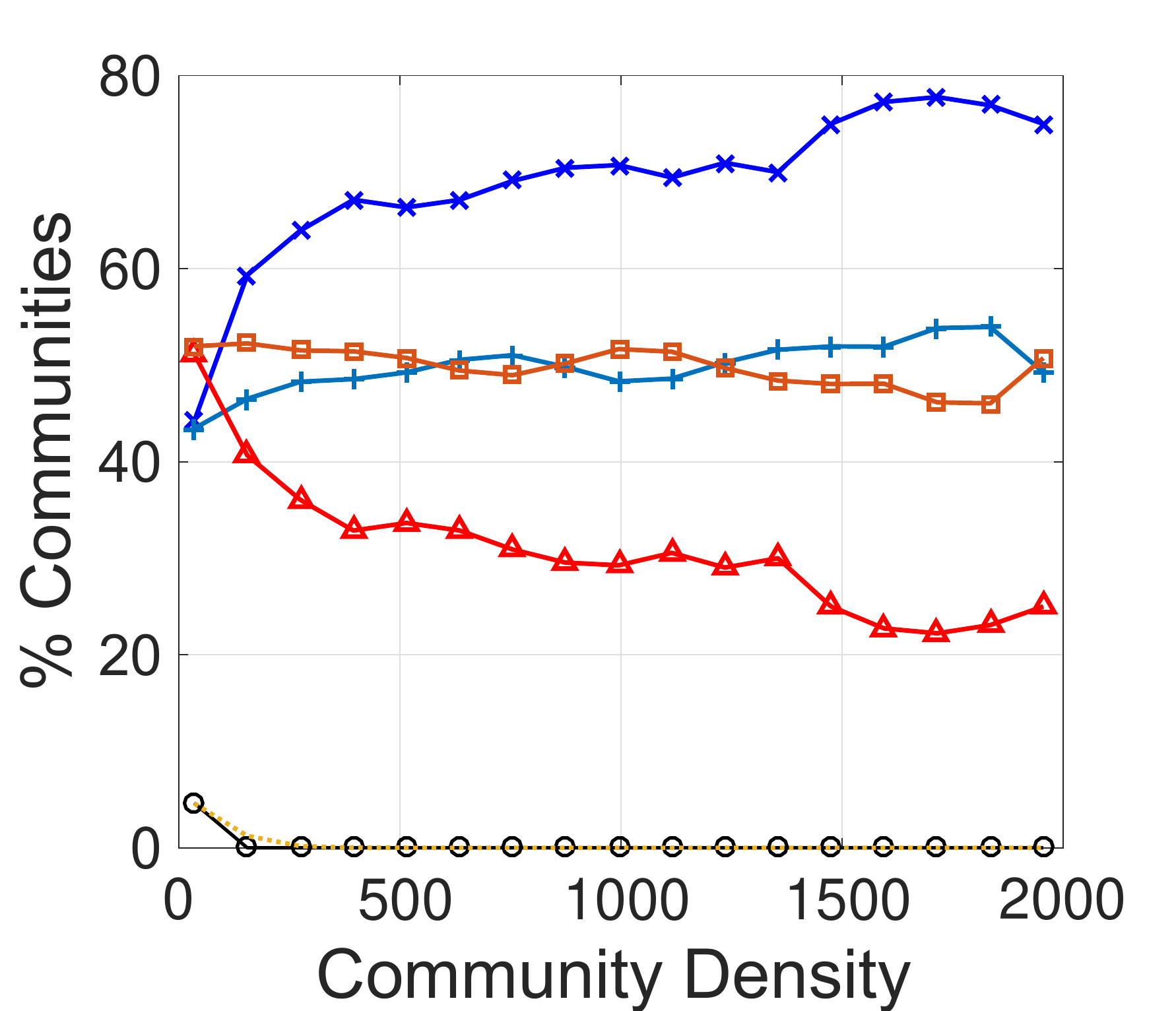}\\
  \end{minipage}}
  \subfigure[Followees]{
  \begin{minipage}{0.275\textwidth}
  	\centering
  	\includegraphics[width=\textwidth]{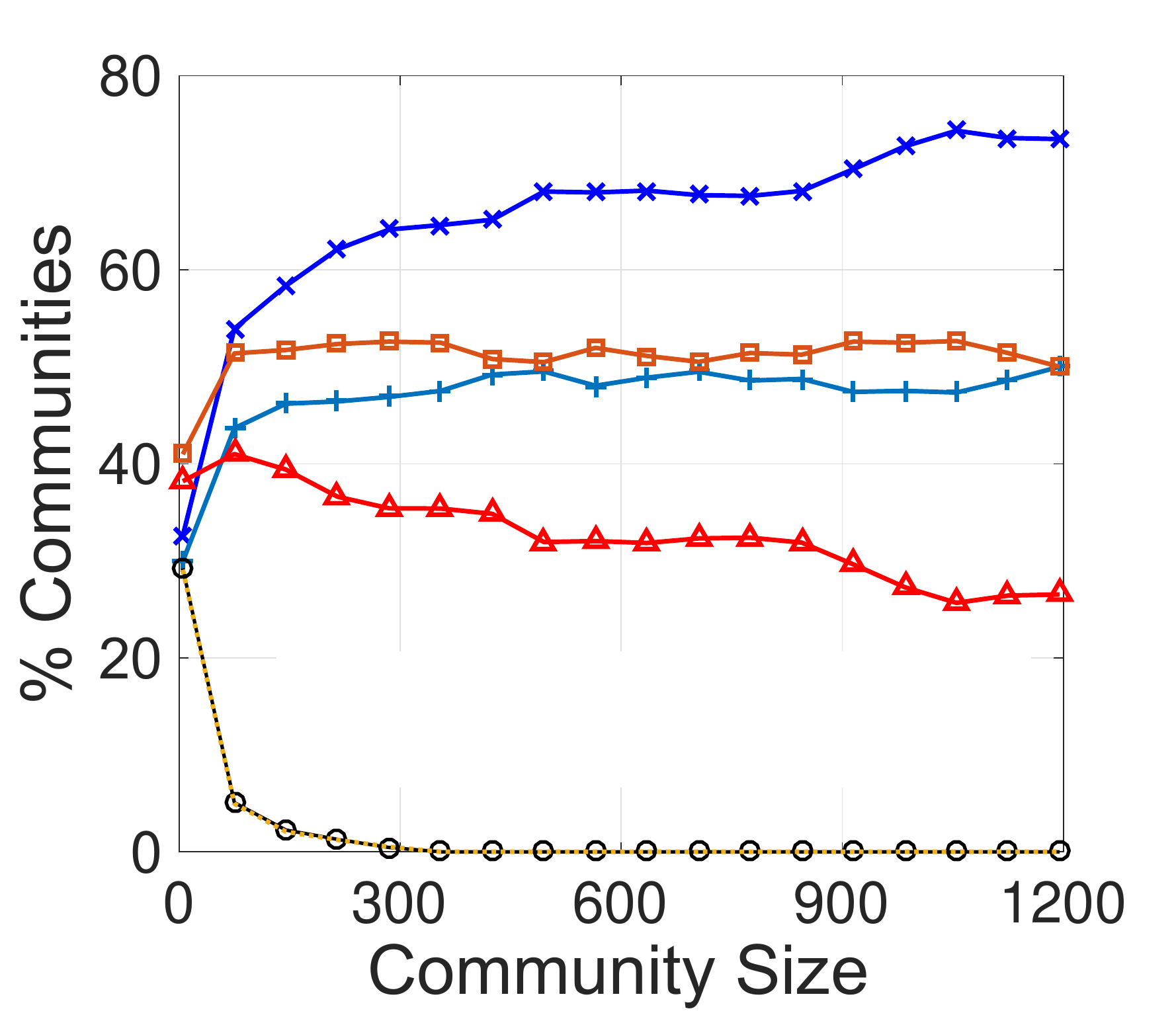}\\
  	\includegraphics[width=\textwidth]{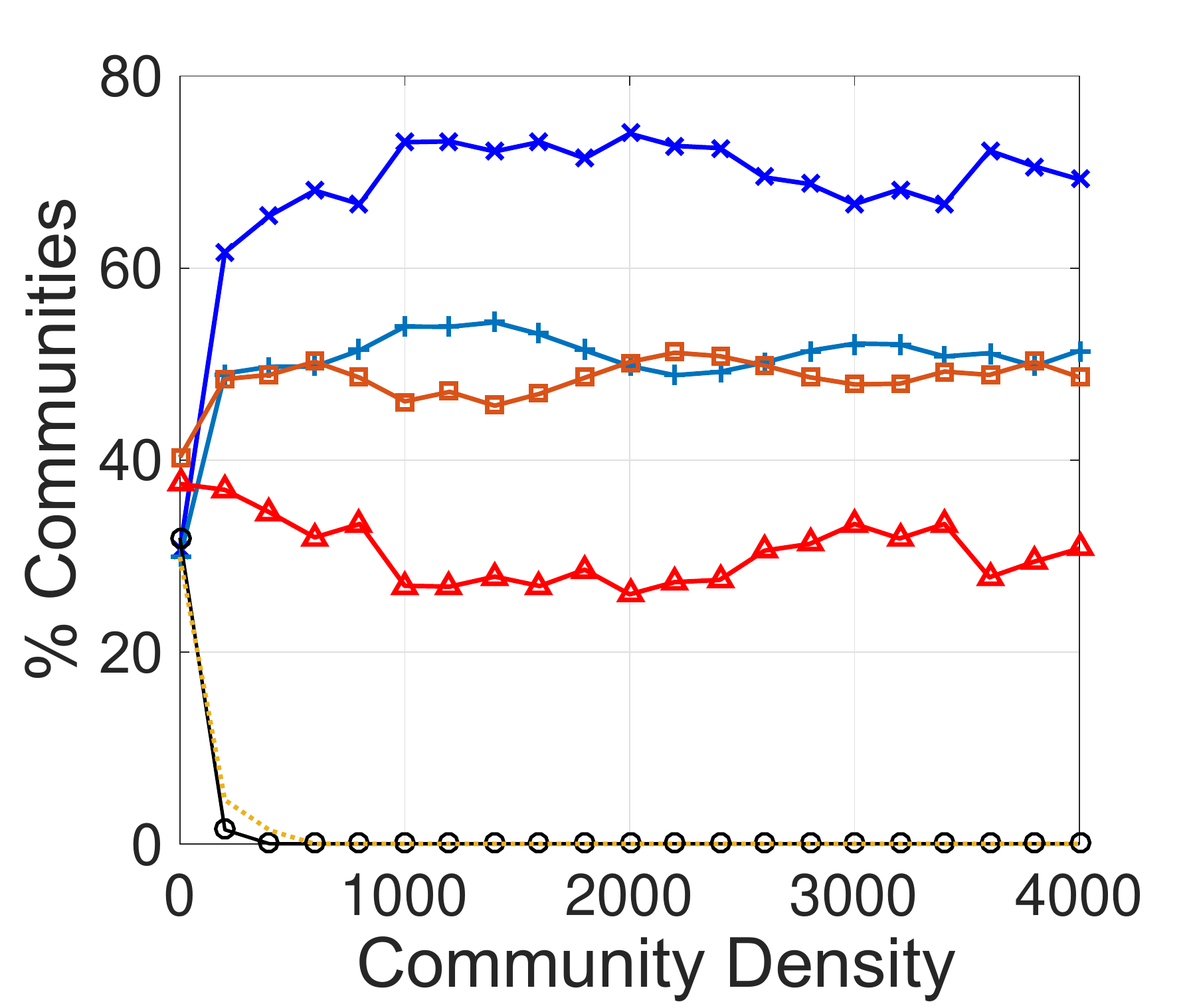}\\
  \end{minipage}}
  \subfigure[Followers]{
  \begin{minipage}{0.275\textwidth}
  	\centering
  	\includegraphics[width=0.96\textwidth]{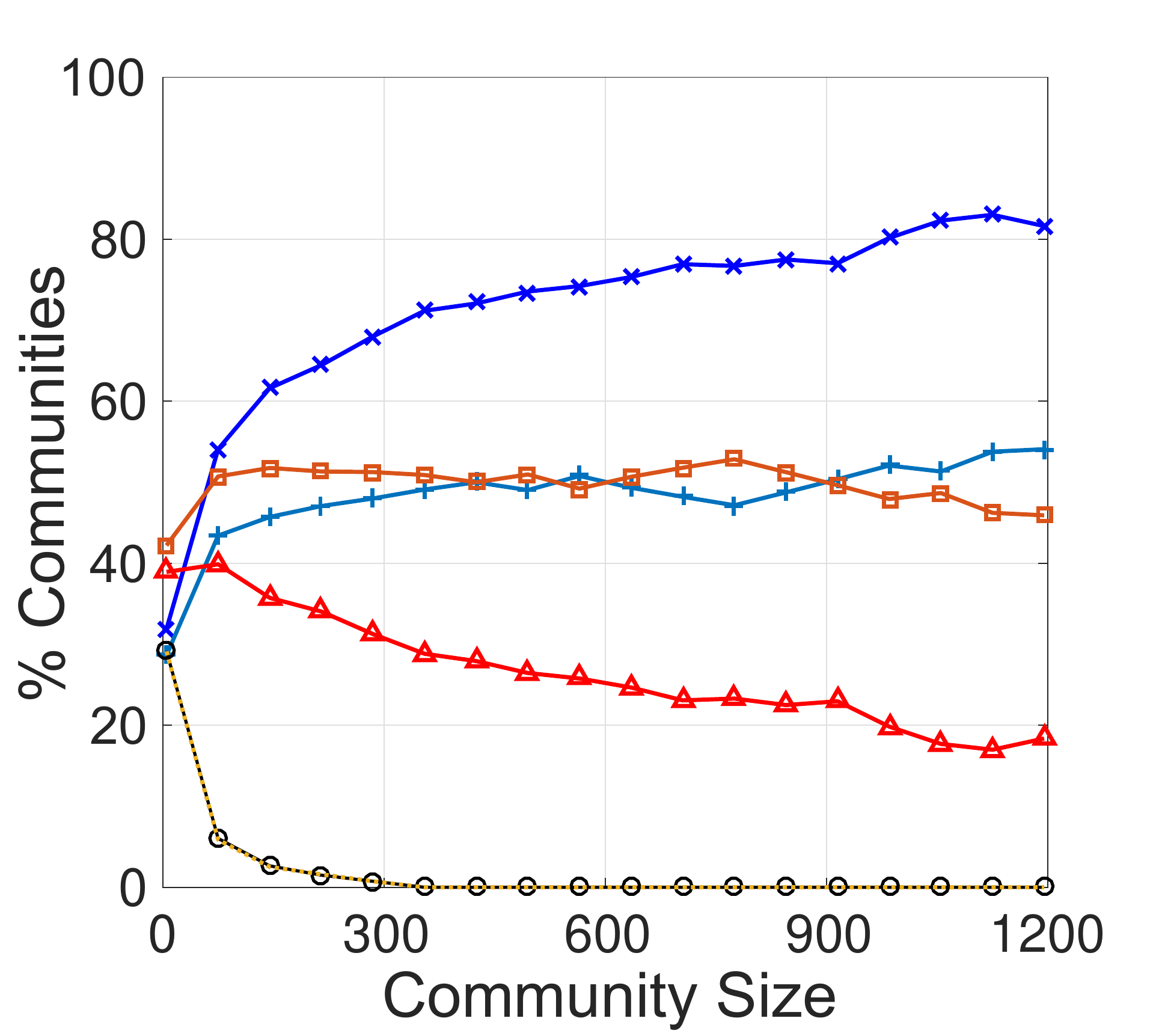}\\
  	\includegraphics[width=0.96\textwidth]{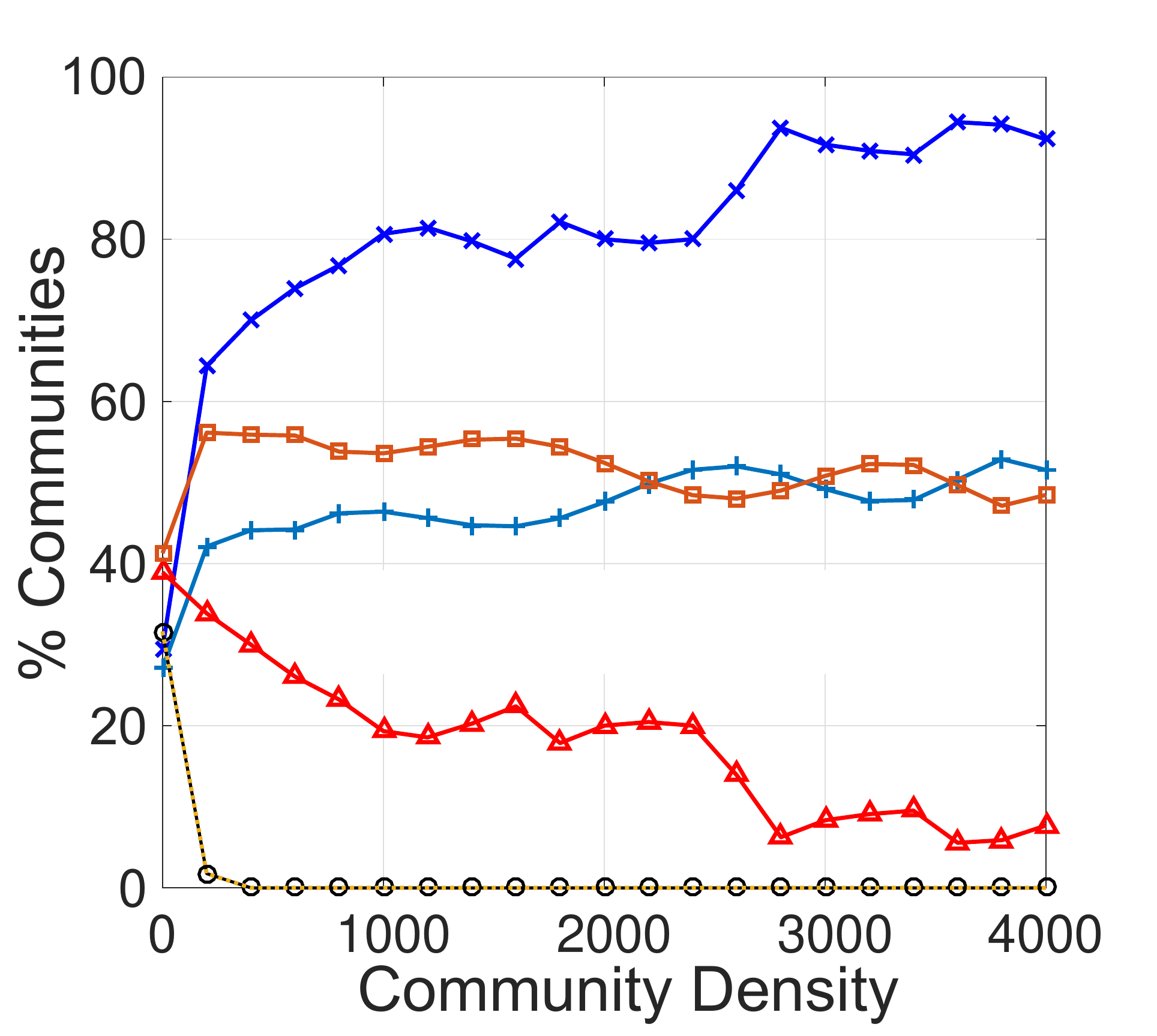}\\
  \end{minipage}}
  \begin{minipage}{0.145\textwidth}
    \centering
  	\includegraphics[width=\textwidth]{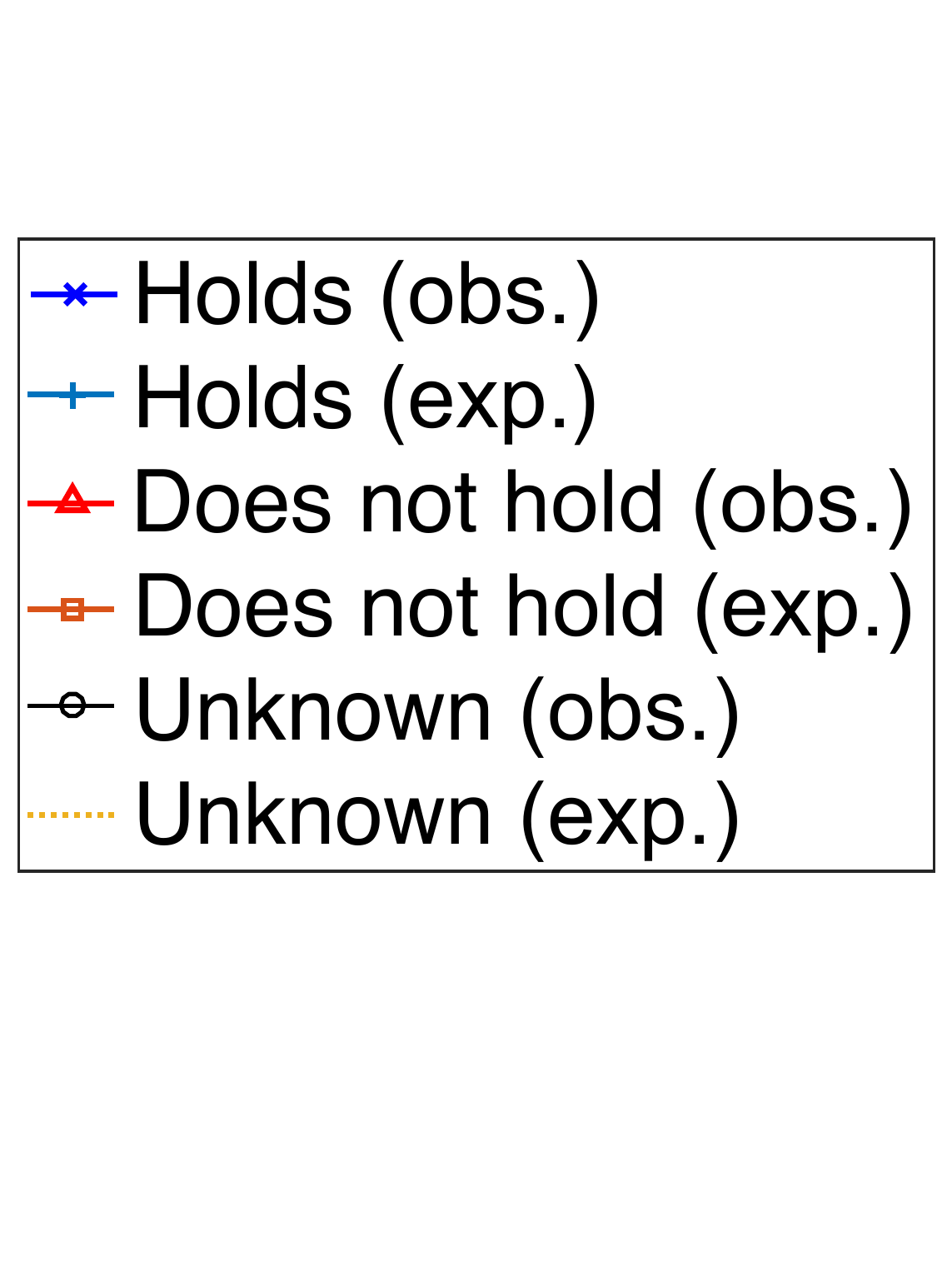}
  \end{minipage}
 \caption{Relations Between Sentiment Paradox and (i) Community Size (upper three), and (ii) Community Density (lower three). The proportion of communities within which the paradox holds becomes larger as communities become larger or denser, ultimately reaching 0.7 (and at times, over 0.9), while the expected magnitude is always around 0.5.}
 \label{fig::msc}
\end{figure*}

\section{Sentiment Paradox in Communities}
\label{sec:communityLevel}
 
Similar to triads, communities also play an important role in understanding social networks~\cite{fortunato2010community}. We take a similar approach to triad-level paradoxes and study sentiments among interacting users within communities. However, we highlight that unlike triads, communities can have different sizes (i.e., number of members) and different levels of interactions among their members (i.e., different densities). Hence, in addition to investigating whether sentiment paradox exists at the community level, we also assess whether the existence or magnitude of such paradoxes depend on the size or level of connections within communities. Similar to triads, we do not expect a sentiment paradox to exist at the community-level as proved in Theorem~\ref{thm::gsp}.

First, we assume user $u_i$, $i=1,2,\cdots,n$ is involved in communities $c_k$, $k=1,2,\cdots,p$. For each $c_k$, we compare the sentiment (i.e., the SWB value) of $u_i$ with the mean and median of that of his connections (friends, followees, or followers) \underline{within the community}. If in a majority of communities that $u_i$ belongs to, $u_i$ is less positive than his connections in the community, we denote $u_i$ is exhibiting the paradox. Finally, we compute the fraction of such users in the network and perform statistical significance analysis. Table~\ref{tab::triad_community} has the results, which we summarize as the following paradox:

\begin{myPara}[Community Sentiment Paradox]
\label{par::CNSP}{Your friends, followees, or followers within a community are more positive than you.}
\end{myPara}

Additionally, we conduct an experiment similar to the one performed to verify the common-neighbor sentiment paradox in the last section, in which we only compare sentiments between users and a subset of their connections with whom users share a community. Statistical significance is computed in the same way as before.

The results are shown in Table~\ref{tab::triad_community}. We observe a sentiment paradox within such users. As users in our dataset mostly form communities around a common interest (one community often refers to a certain topic), we denote this paradox as the \textit{common-interest sentiment paradox}:

\begin{myPara}[Common-interest Sentiment Paradox]
\label{par::community3}
Your friends, followees, or followers with whom you share some interests are more positive than you.
\end{myPara}

\vspace{0.5em}
\noindent \textbf{Impact of Community Variations.} To assess the impact of variations in communities on the sentiment paradox, we measure the paradox magnitude by changing the \emph{community size} (i.e., number of members/nodes) or \emph{community density} (i.e., number of connections/edges). We vary the community size from 1 to 1,200, and community density from 1 to (i) 2,000 in the undirected network, and (ii) 4,000 in the directed network.\footnote{Community size and density both follow a power-law-like distribution. Only around ten (less than 0.005\%) communities exist in which number of users is greater than 1,200, or friendships among users is greater than 2,000, or following and follower relationships among users is greater than 4,000.} Then, we calculate the paradox magnitude within these communities.
The results are in Figure~\ref{fig::msc}.

We observe from Figure~\ref{fig::msc} that the proportion of communities within which the paradox holds becomes larger as communities become larger or denser,  ultimately reaching 0.7 (and at times, over 0.9), while the expected magnitude is always around 0.5. Even when the community size or its density is very small, the observed proportion of communities where the sentiment paradox holds is always higher than the expected proportion, and the observed proportion of communities where the sentiment paradox does not hold is always lower than the expected values.

\section{Connections to Network Paradoxes}
\label{sec:factors}

Social networks exhibit many counter-intuitive properties. We assess the connection between sentiment paradox and two of the most commonly observed network paradoxes: (1) friendship paradox and (2) activity paradox, which provides opportunities to investigate the relationships among user sentiments, social connections and activities.

\begin{figure}[t]
 \centering
 \includegraphics[width=.48\columnwidth]{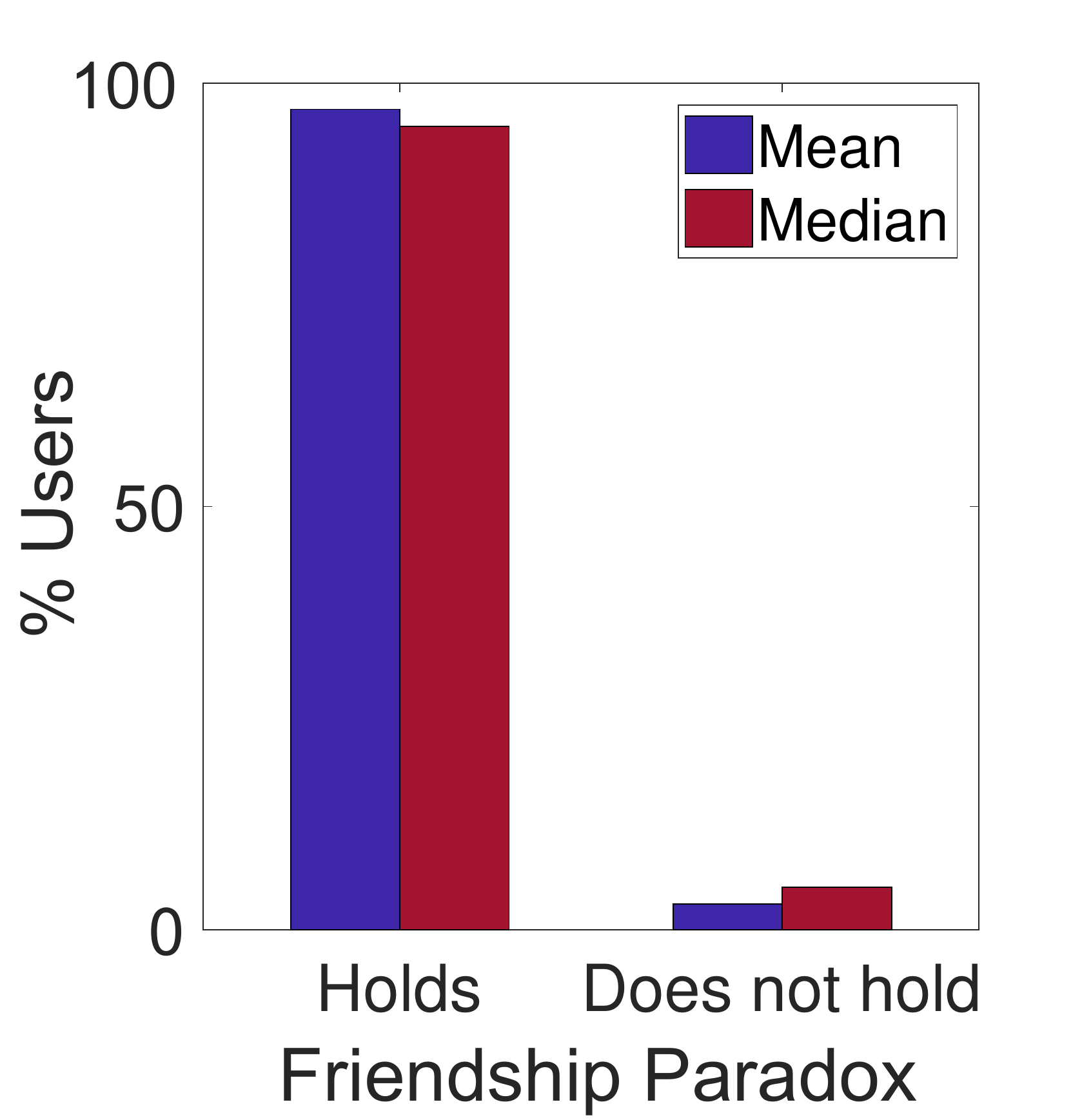}
 \includegraphics[width=.48\columnwidth]{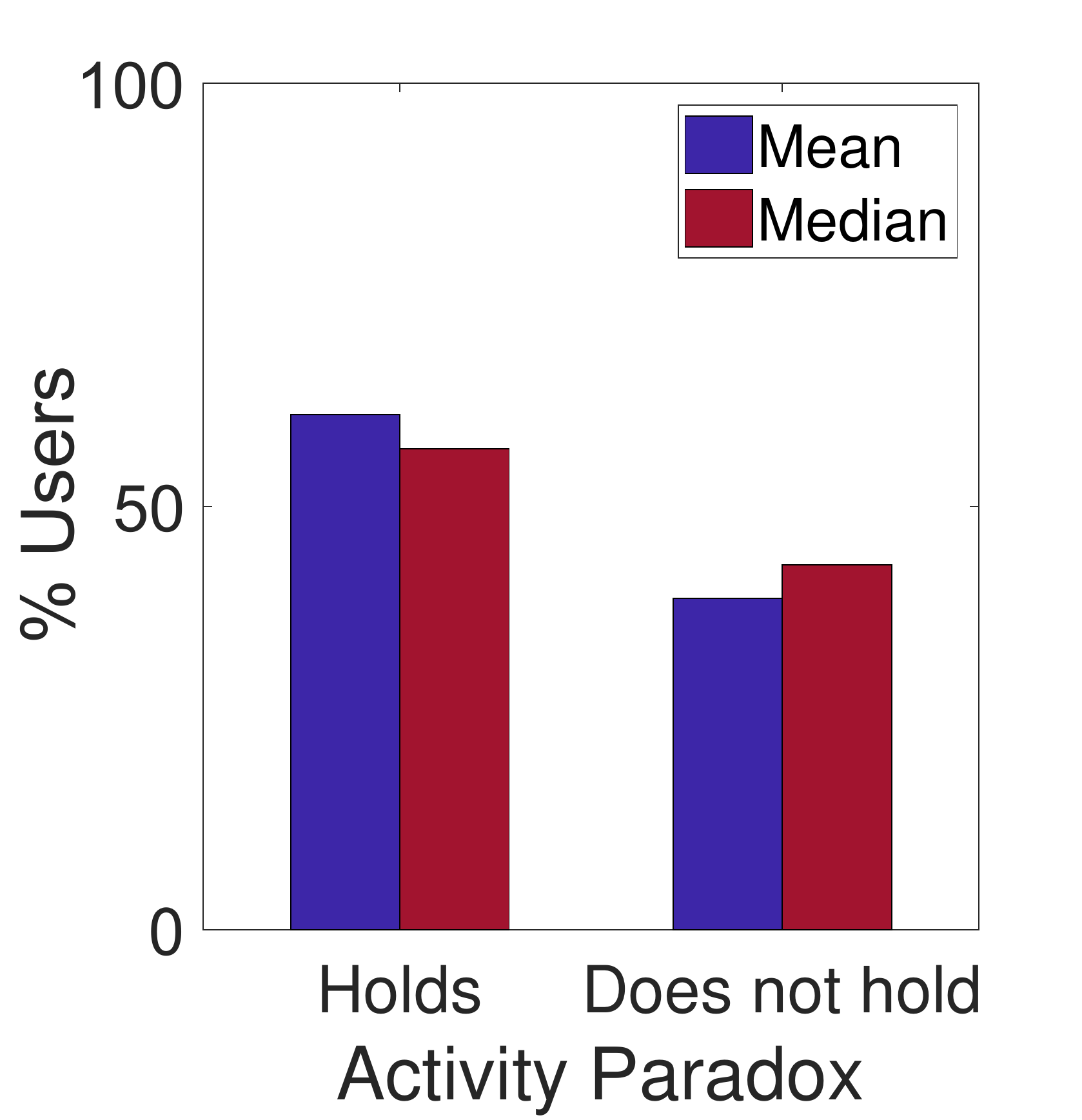}
   \caption{Friendship Paradox (left) and Activity Paradox (right).\protect\footnotemark[5] Both paradoxes hold for a majority of users and thus exist in the network.}
   \label{fig::otherParadox}
\end{figure}
\footnotetext[5]{Both network paradoxes are \underline{expected to exist} based on the mean value as the distributions of node degrees and user activity exhibit a heavy tail~\cite{kooti2014network}, which are different from that of user sentiments.}

\subsection{Friendship Paradox}

One of the most well-known network paradoxes is the friendship paradox, first observed by Feld~\cite{feld1991your}, which states that users have fewer friends than their friends, on average. 
The paradox also holds for the median value~\cite{kooti2014network}.
In our data, in addition to sentiment paradoxes at different network levels, we observe a friendship paradox for most users (both mean and median, see Figure~\ref{fig::otherParadox}).
Here, we explore the interplay between node degrees and the sentiment paradox, motivated by the following facts:
\begin{enumerate}
\item A user with degree $d$ contributes his SWB value $d$ times to the average SWB distribution of friends of users. We illustrate this fact using an example.
\begin{myExp}
\label{exp::e1}
Consider a simple undirected friendship network (see Figure~\ref{exp1}) with four users $u_1$, $u_2$, $u_3$, and $u_4$, whose corresponding SWB values are $+0.1$, $-0.2$, $-0.3$, and $-0.4$. For user $u_1$, the average SWB value of his friends is $\frac{(-0.2)+(-0.3)+(-0.4)}{3}$. The average SWB values of the friends of $u_2$, $u_3$, and $u_4$ are all $+0.1$. Thus, user $u_1$ with degree three contributes his SWB value three times (as a friend of $u_2$, $u_3$, and $u_4$) to the average SWB distribution of friends of users, while other users, with degree one, contribute their SWB values only once.
\end{myExp}

\begin{figure}[t]
    \centering
    \includegraphics[width=0.6\columnwidth]{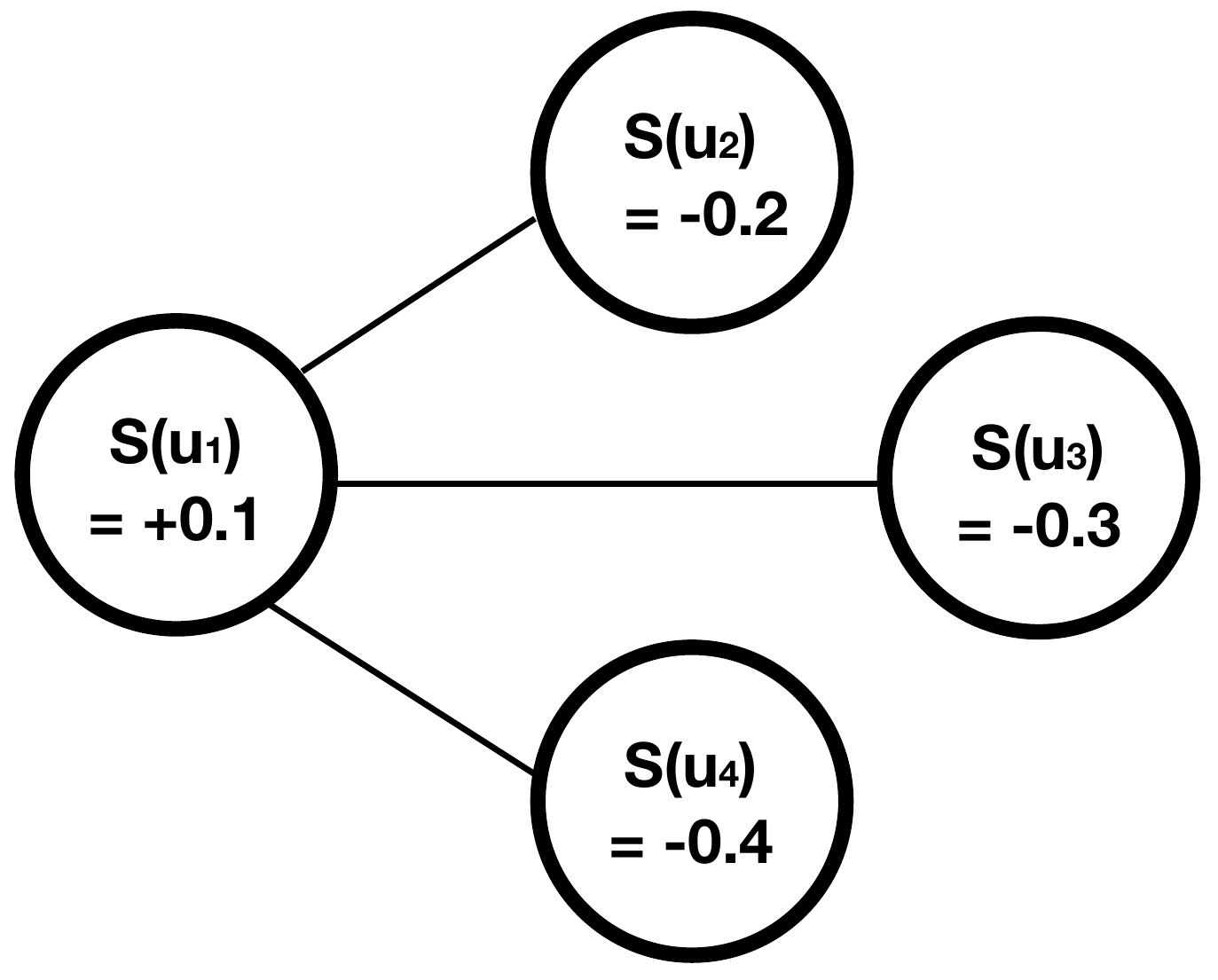}
 \caption{An Illustration for Example 1}
 \label{exp1} 
\end{figure}

\item Compared with the distribution of user sentiments (SWBs, see Figure~\ref{fig::swbDist}), the distributions of the mean and median of sentiments of friends, followees or followers of users are skewed to the right (see Figure~\ref{fig::distSkewing} for friends), i.e., the latter distributions are a weighted version of the former one, weighting those with comparatively high SWB values more.
\end{enumerate}

\begin{figure}[t]
 \centering
 \includegraphics[width=0.49\columnwidth]{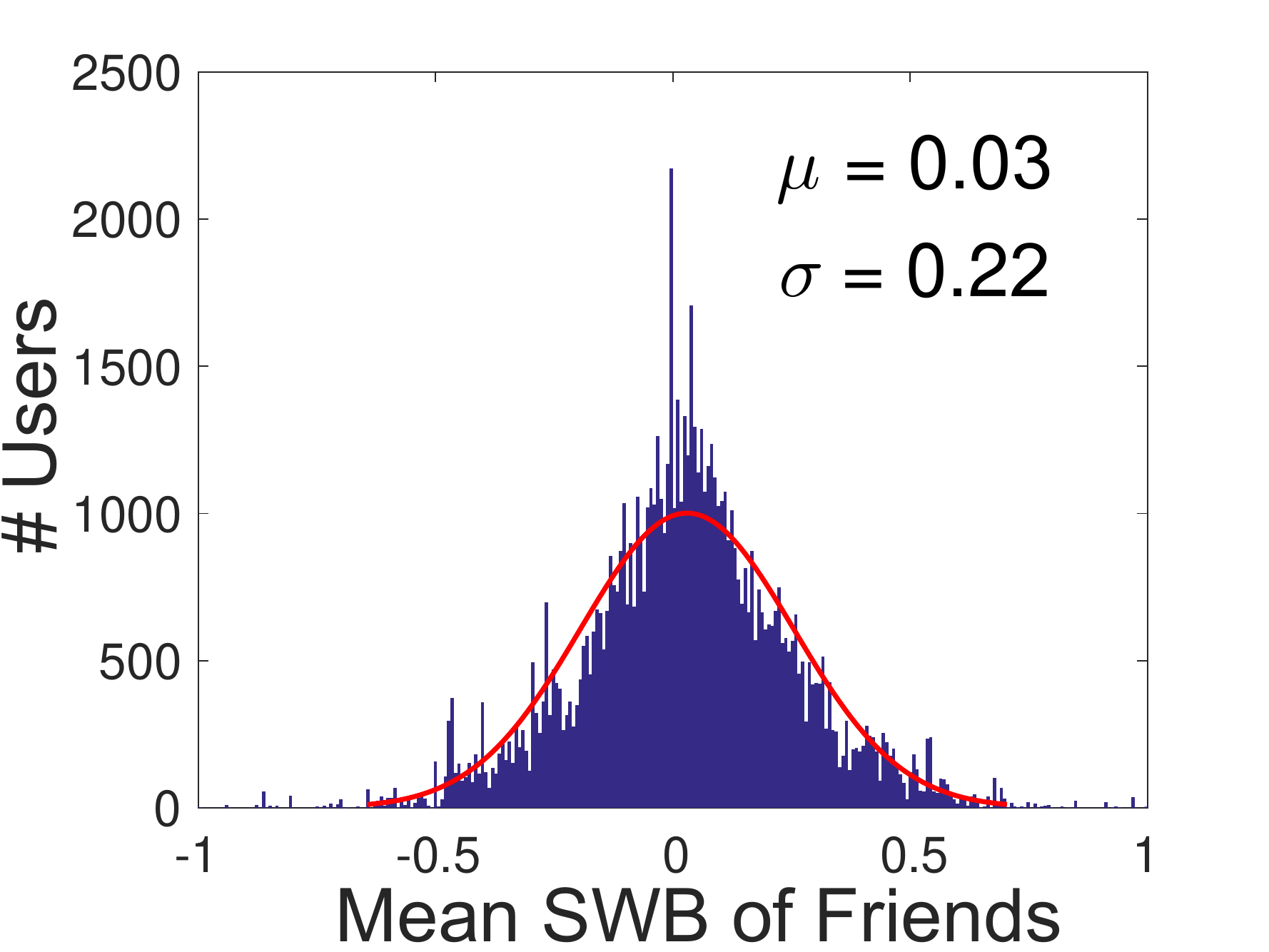}
 \includegraphics[width=0.49\columnwidth]{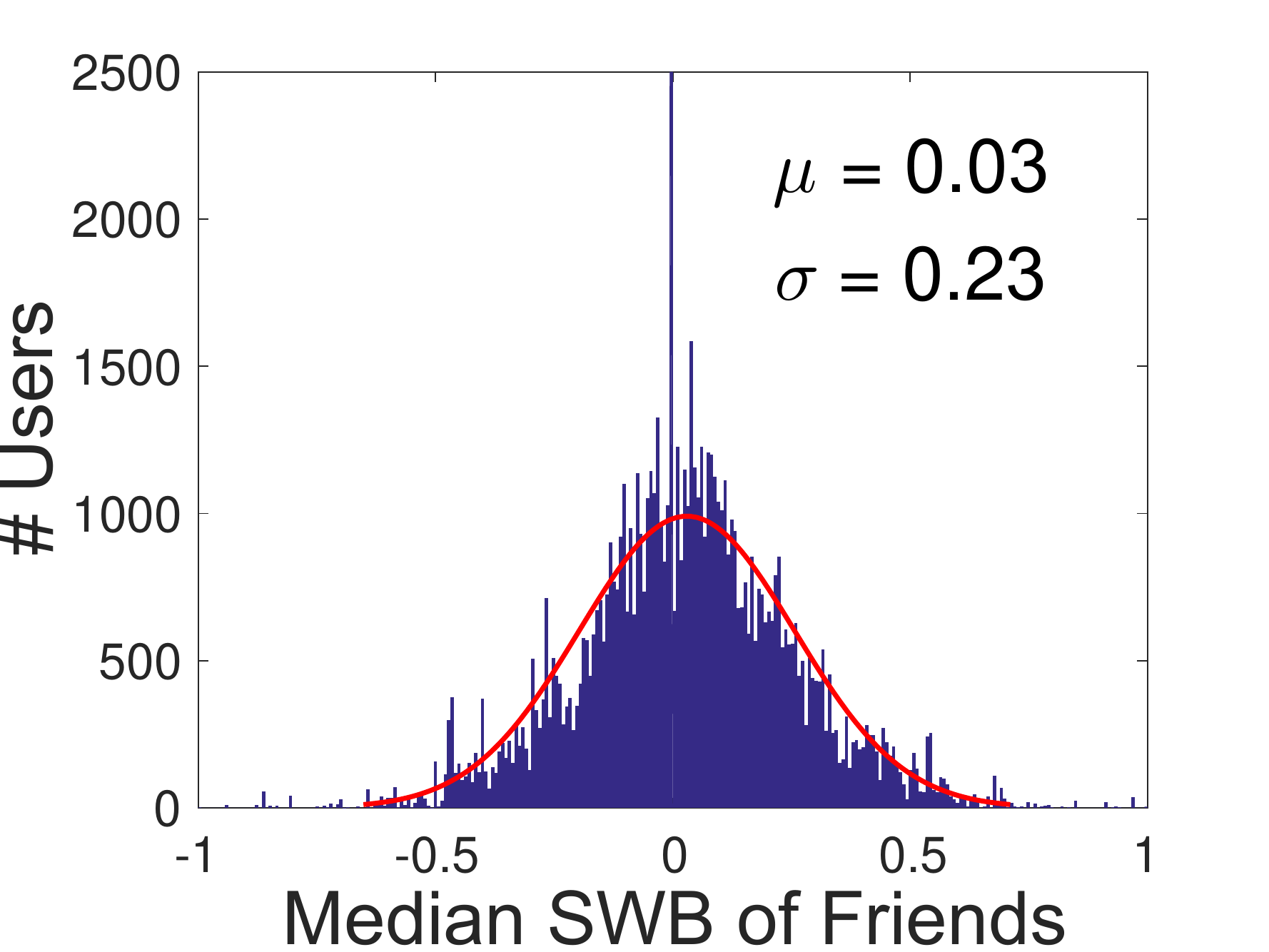}
   \caption{Distribution of Mean of Friend Sentiments (left), and Median of Friend Sentiments (right). Compared to the distribution of user sentiments (SWBs, see Figure~\ref{fig::swbDist}), the distributions of the mean and median of sentiments of friends of users are skewed to the right, i.e., $\mu$ increases.}
 \label{fig::distSkewing}
\end{figure}

Given these two facts, it is natural to study whether users with relatively high (in-, out-) degrees are more positive (i.e., have higher SWB values) than those with relatively low (in-, out-) degrees. We verify this hypothesis in two ways.

\begin{table}[t]
\centering
\caption{Correlations Between User Sentiments (SWB) and Number of Social Connections. Correlations are all positive and highly significant as $p$-values approach zero.}
\label{tab::correlation}
\begin{tabular}{lc}
\toprule[1pt]
 & \textbf{Correlation Coefficient} \\ \hline
\textbf{(SWB, \# Friends)} & 0.05 ($p$-value $\rightarrow 0$)\\
\textbf{(SWB, \# Followees)} & 0.04 ($p$-value $\rightarrow 0$)\\
\textbf{(SWB, \# Followers)} & 0.04 ($p$-value $\rightarrow 0$)\\
\bottomrule[1pt]
\end{tabular}
\end{table}

\begin{table}[t]
\centering
\caption{Average Number of Social Connections for Positive, Negative and Neutral Users. In general, positive users have (30\% to 45\%) more social connections than the others.} 
\label{tab::userDegree}
\subtable[When SWB values of users are between $-1$ and $+1$]{
\label{subtab::full}
\begin{tabular}{lccc}
\toprule[1pt]
\textbf{Users} & \textbf{Friends} & \textbf{Followees} & \textbf{Followers} \\ \hline 
{Positive ($+$)} & 5.09             & 8.01               & 8.16               \\ 
{Negative ($-$)} & 3.58             & 5.99               & 5.88               \\ 
{Neutral (0)}  & 3.70             & 5.90               & 5.80                \\ 
\hline
{Overall}       & 4.25             & 6.88               & 6.88               \\ 
\bottomrule[1pt]
\end{tabular}}
\subtable[{{When SWB values are between $-0.5$ and $+0.5$}}]{
\label{subtab::partial}
\begin{tabular}{lccc}
 \toprule[1pt]
\textbf{Users}  & \textbf{Friends} & \textbf{Followees} & \textbf{Followers} \\ \hline 
{Positive ($+$)} & 5.05             & 7.94               & 8.07               \\ 
{Negative ($-$)} & 3.67             & 6.06               & 5.96               \\ 
{Neutral (0)}  & 3.70             & 5.90               & 5.80               \\ \hline
{Overall}       & 4.27             & 6.87               & 6.88               \\ \bottomrule[1pt]
\end{tabular}}
\end{table}

\begin{figure*}[t]
 \centering
     \includegraphics[width=0.32\textwidth]{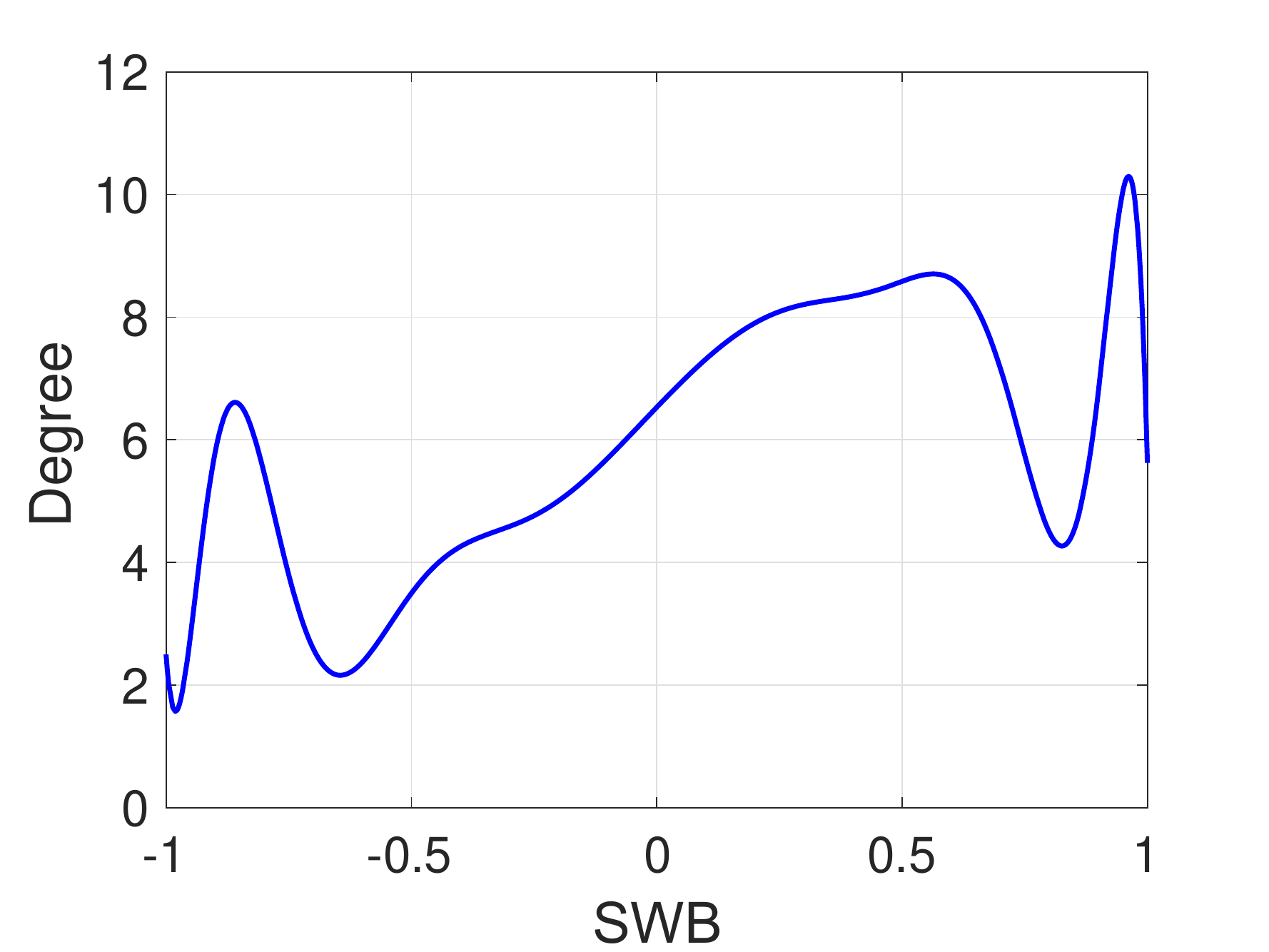} 
     \includegraphics[width=0.32\textwidth]{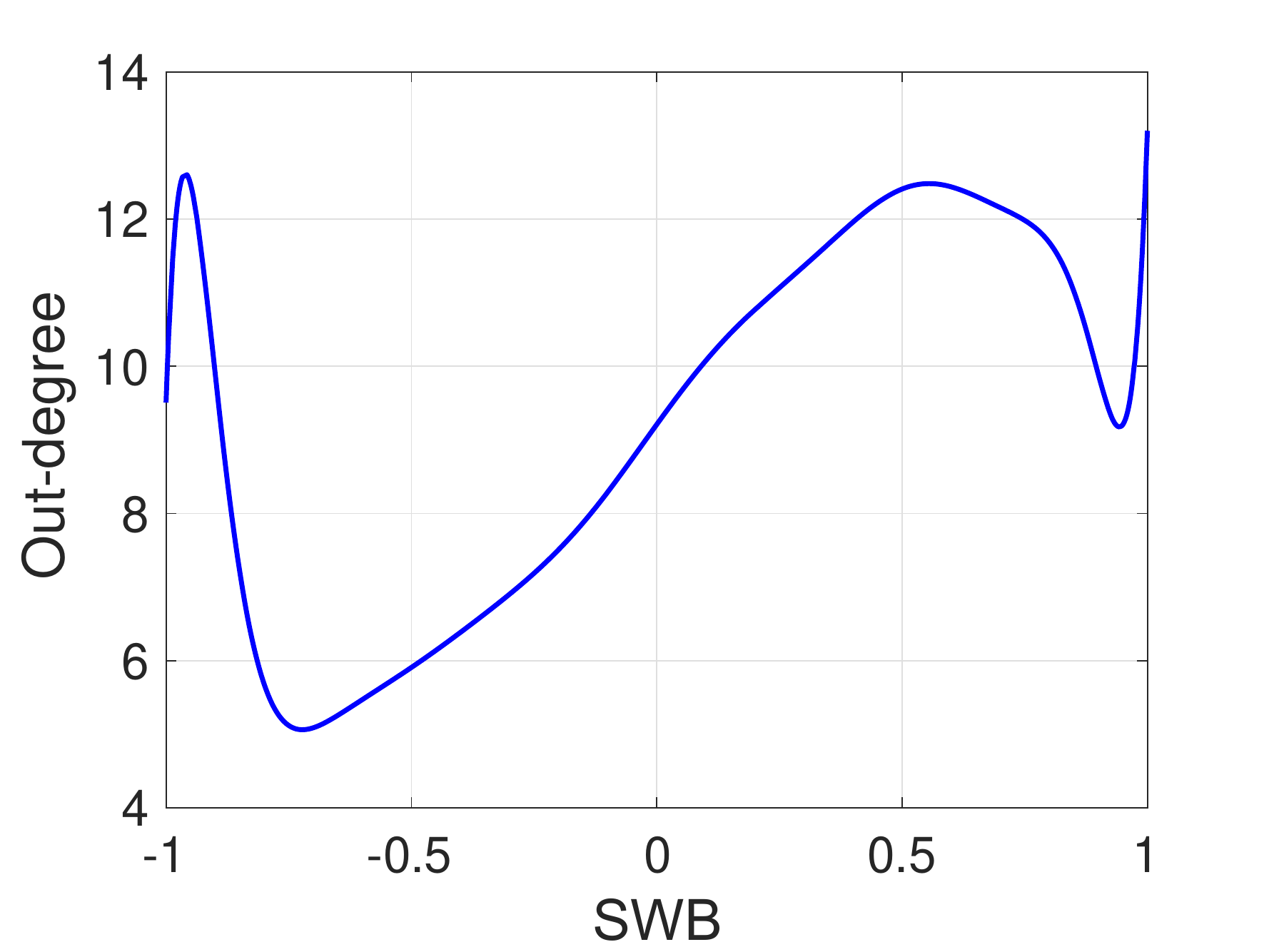} 
     \includegraphics[width=0.32\textwidth]{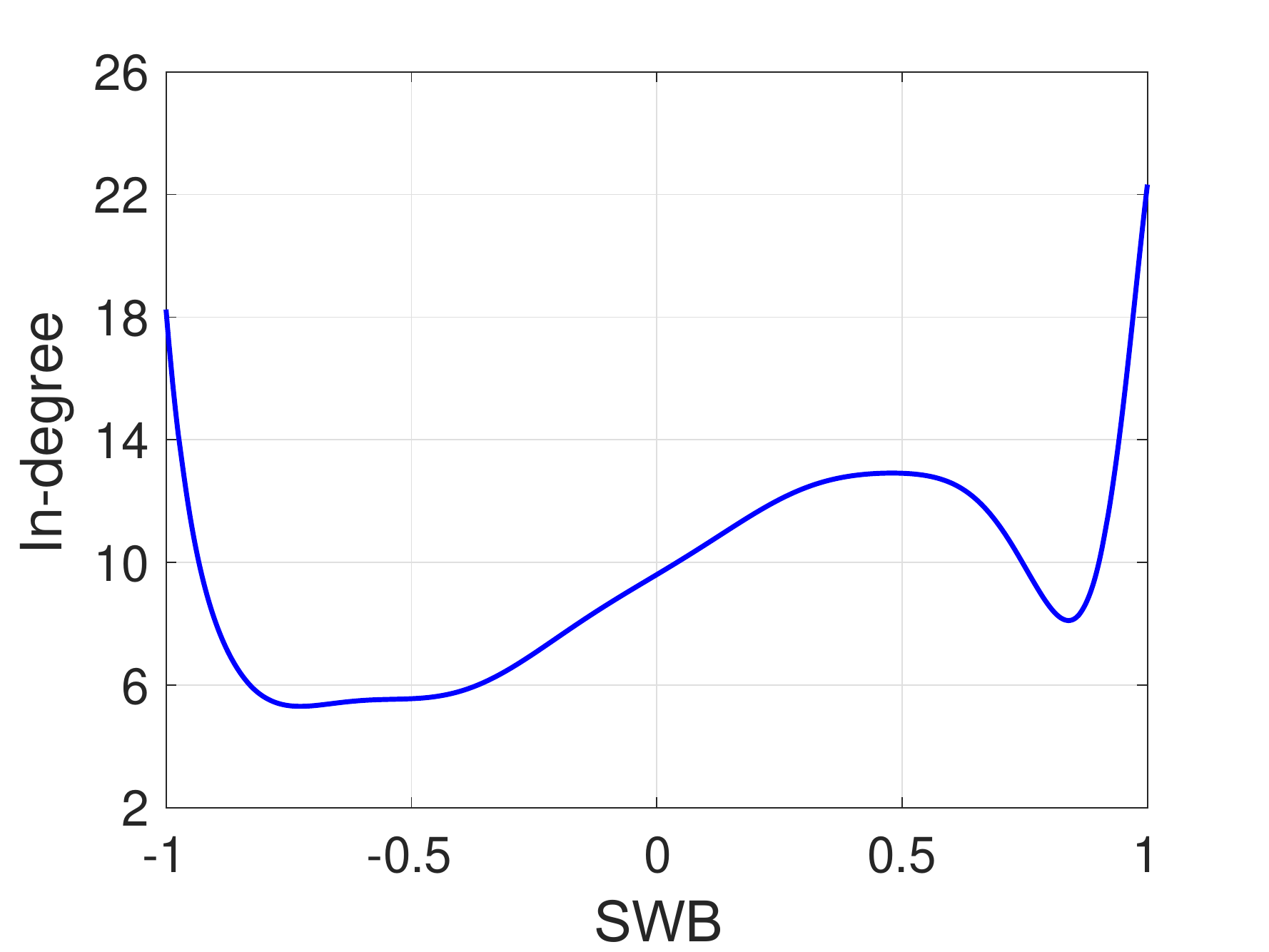} 
   \caption{Relations Between SWB and (i) Degrees (left), (ii) Out-degrees (middle), and (iii) In-degrees (right) of Users. SWB values of users have positive relations with (in-, out-) degrees, in particular, when SWB values are between $-0.5$ and $+0.5$.}
 \label{fig::swbDegree}
\end{figure*}

\vspace{0.5em}
\noindent \textbf{I.} Without labeling a user as positive, negative or neutral, we directly compute the correlation coefficient between user sentiments (SWB values) and their number of (i) friends (degrees), (ii) followees (out-degrees), and (iii) followers (in-degrees). Results are presented in Table~\ref{tab::correlation}, which indicate that the sentiments of users are positively correlated to their number of friends, followees and followers with $p$-values approach zero (i.e., results are highly significant).

We further visualize such correlations, where a least-square fit of the trend is provided in Figure~\ref{fig::swbDegree}. It further validates that the SWB value of users is positively related to their (in-, out-) degrees, especially when SWB values are between $-0.5$ and $+0.5$. Concretely, the (in-, out-) degree of users with SWB value $+0.5$ are about six more than that of users with SWB value $-0.5$. In other words, more positive users usually have more friends, followees and followers. 

Note that the group of users with extreme sentiments (i.e., whose SWB values approach $-1$ or $+1$) are not representative enough as they occupy a very small proportion (less than 7\%) in the population. In Figure~\ref{fig::swbDegree}, it can be observed that such users seem to have significantly larger degrees. However, such phenomenon can be attributed to the degree distribution, which is almost \textit{power-law} and has a heavy tail. Once one user in the group has a significantly larger degree, it easily leads to a peak when using the least-square fit. Hence, our conclusion here is obtained mainly based on users with SWB values between $-0.5$ and $+0.5$ as these users are more representative than users whose sentiments are not in this range.
\vspace{1mm}

\noindent \textbf{II.} We further consider all users and group them based on being positive ($S(u)>0$), negative ($S(u)<0$), or neutral ($S(u)=0$). The average (in-, out-) degree for users within each group is then calculated and provided in Table~\ref{subtab::full}.  Table~\ref{subtab::full} shows that positive users have more friends, followees, and followers compared to the other users, which is also above the averages computed for all users. In particular the number of friends, followees, and followers of positive users are on average 30\% to 45\% greater than that of negative users. Additionally, we also compute the average (in-, out-) degree of users whose SWB values are between $-0.5$ and $+0.5$. The results are shown in Table~\ref{subtab::partial}, which leads to the same conclusion.

\begin{figure}[t]
 \centering
   \includegraphics[width=.68\columnwidth]{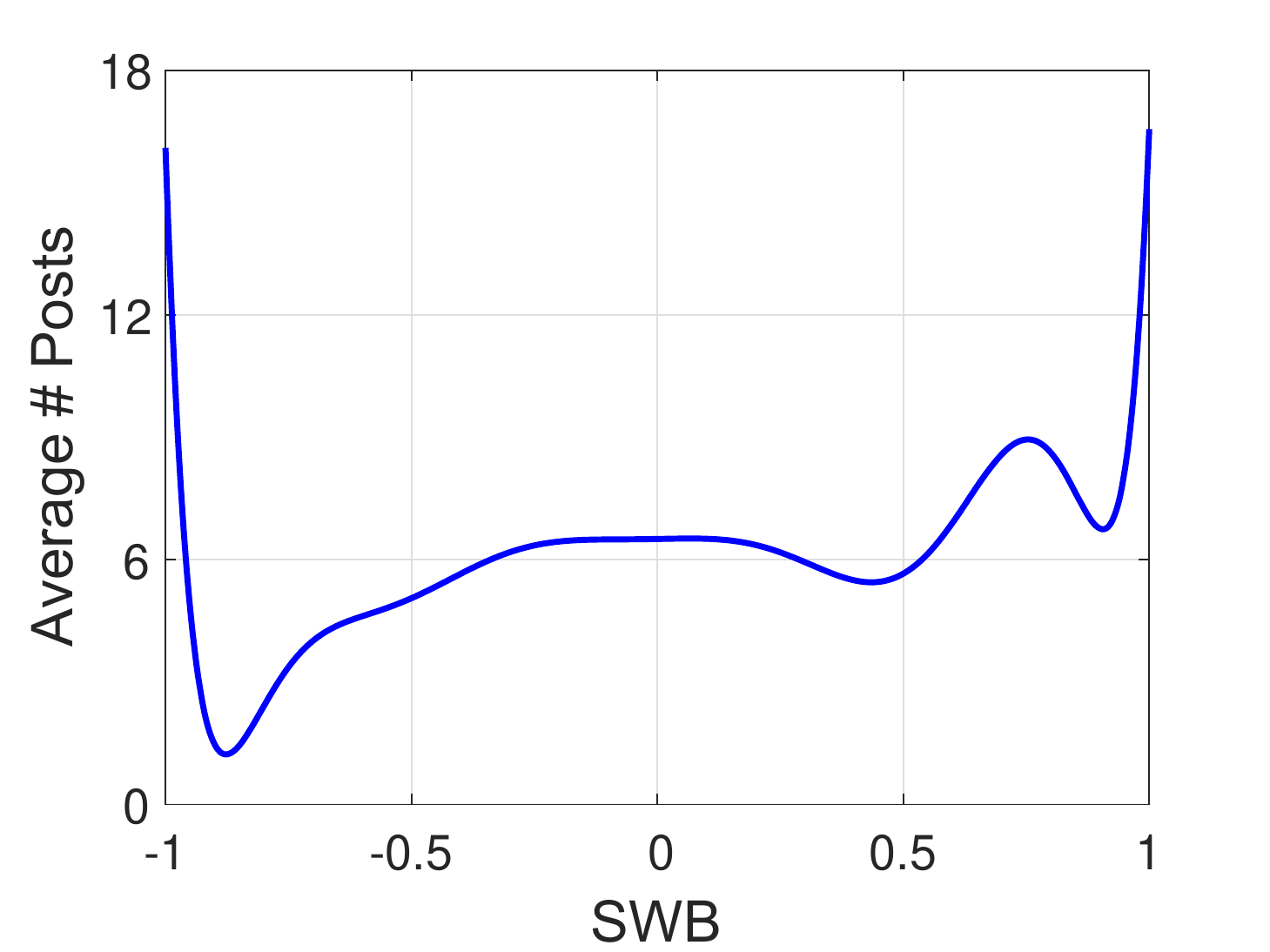} 
   \caption{Relationship Between User Sentiments (i.e., SWB) and Activity. User activity (i.e., the average number of user posts per 30 days) is rarely affected by user SWB values between $-0.5$ and $+0.5$.}
 \label{fig::swbActivity}
\end{figure}

\begin{table*}[t]
\Huge
\centering
\caption{Feature List}
\label{tab::feature}
\resizebox{\textwidth}{!}{
\begin{tabular}{ll}
\toprule[2pt] 
\textbf{Feature Group (\# Features)} & \textbf{Description} \\ \midrule[1pt]
{General Sentiment Paradox (6)} & 
Mean and median of SWB values of one's social connections (friends, followees, or followers) \\
{Triad Sentiment Paradox (6)} & 
Mean and median of SWB values of one's social connections in a triad \\
{Common-neighbor Sentiment Paradox (6)} & 
Mean and median of SWBs of one's social connections with whom he shares common neighbors  \\
{Community Sentiment Paradox (6)} & 
Mean and median of SWB values of one's social connections in a community \\
{Common-interest Sentiment Paradox (6)} & 
Mean and median of SWBs of one's social connections with whom he shares common interests  \\
{{Friendship Paradox (9)}}  & The number of degrees, in-degrees and out-degrees of oneself \& \\ 
 & Mean and median of degrees, in-degrees and out-degrees of one's social connections \\
 \bottomrule[2pt]
\end{tabular}}
\end{table*}

\subsection{Activity Paradox}
In social networks such as {Twitter}~\cite{hodas2013friendship} and {Digg}~\cite{kooti2014network}, researchers have discovered the existence of an activity paradox: users are less active than their friends, on average. We observe an activity paradox, less strongly than friendship paradox, in our data (see Figure~\ref{fig::otherParadox}), which inspires us to explore the potential relationships between user activity and sentiments. We quantify user activity as follows:

\begin{myDef}[User Activity]
\label{def::userActivity}
Suppose user $u$ has posted $n$ posts in a social network, where the first post was published on date $d_1$ and the last one was posted on date $d_n$. The activity of user $u$ is defined as
\begin{equation}
A(u) = \frac{\Delta{t}}{d_n-d_1}n,
\end{equation}
where $\Delta{t}$ is size of time window where we measure activity. 
\end{myDef}

Note that $\frac{d_n-d_1}{\Delta{t}}$ indicates how many $\Delta{t}$'s (e.g., months) a user has been active on the network and $A(u)$ indicates the average number of posts of user $u$ in $\Delta{t}$ period.

The relation between user sentiments (i.e., SWB	values) and activity (i.e., the average number of posts of users per $\Delta{t}$=30 days) is shown in Figure~\ref{fig::swbActivity}. We observe that the value of $\Delta{t}$ does not influence the result. There seems to be a slight positive correlation; however, the number of user posts is rarely affected by the user SWB values if between $-0.5$ and $+0.5$, which covers 93\% of our users. Therefore, we do not consider user activity to have a significant impact on sentiment paradox.

\begin{table}[t]
    \centering
    \caption{Distribution of Positive, Negative and Neutral Users}
    \label{tab::userSentiments}
    \begin{tabular}{lrr}
    \toprule[1pt]
       \textbf{User}         & \textbf{Number}  & \textbf{Proportion} \\ \hline
       Positive ($+$) & 50,705  & 43.92\%  \\
       Negative ($-$) & 61,066  & 52.90\%  \\
       Neutral ($0$)  & 3,673   & 3.18\%   \\ \hline
       Total        & 115,444 & 100.00\% \\ \bottomrule[1pt]
    \end{tabular}
\end{table}

\section{User Sentiment Prediction}
\label{sec:application}

Sentiment paradoxes reveal a certain relationship between users and their social connections (friends, followees, or followers) at triad-, community- and network-level, i.e., in general, users are less positive than their social connections. In this section, we demonstrate how a user's sentiment (positive or negative) can be predicted by investigating the general sentiments of his social connections at triad-, community- and network-level. Before the elaboration, we provide the distribution of positive, negative and neutral users in Table~\ref{tab::userSentiments}.

To predict user sentiments (positive or negative), we regard it as a binary classification problem to be addressed within a supervised machine learning framework. Specifically, we represent each user as a set of machine learning features. Features are inspired by the validated five sentiment paradoxes, and the friendship paradox which has been validated to be correlated to the sentiment paradox. Features are presented in Table~\ref{tab::feature}. Then, several common supervised classifiers are trained and used to predict user sentiments (positive or negative) based on ten-fold cross-validation. Results are evaluated by accuracy (ACC) and AUC value.

Table~\ref{tab::results} provides the overall results. Results are obtained by using XGBoost~\cite{chen2016xgboost}, which performs best among supervised classifiers including logistic regression, decision trees, na\"ive Bayes, random forests, and Support Vector Machine (SVM) - see Table~\ref{tab::classifiers} for their performance comparison. Results in Table~\ref{tab::results} indicate that (1) among single sentiment paradoxes, the general one performs best in predicting user sentiments; (2) when combining all sentiment paradoxes, it outperforms when separately using single ones; and (3) in general, using all features (five sentiment paradoxes plus the correlated friendship paradox) perform best, which can achieve 62\% accuracy ratio and 60\% AUC value.

\begin{table}[t]
\centering
\caption{Performance Comparison by using Various Supervised Classifiers in Predicting User Sentiments. XGBoost performs best among all selections.}
\label{tab::classifiers}
\begin{tabular}{lcc}
\toprule[1pt]
\textbf{Classifier} & \textbf{ACC} & \textbf{AUC} \\ \hline
{Logistic Regression} & .613 & .590 \\
{Decision Tree} & .596 & .580 \\
{Na\"ive Bayes} & .600 & .580 \\
{Random Forest} & .590 & .580 \\ 
{SVM} & .587 & .573 \\
{XGBoost} & .620 & .600 \\
\bottomrule[1pt]
\end{tabular}
\end{table}

\begin{table}[t]
\huge
\centering
\caption{Performance Comparison by using Various Feature Groups in Predicting User Sentiments. Among all single sentiment paradoxes, the general one performs best. When combining all sentiment paradoxes outperforms when separately using single ones.}
\label{tab::results}
\resizebox{\columnwidth}{!}{
\begin{tabular}{lcc}
\toprule[1.5pt]
\textbf{Feature Group} & \textbf{ACC} & \textbf{AUC} \\ \hline
{General Sentiment Paradox} & .601 & .581 \\
{Triad Sentiment Paradox} & .589 & .568 \\
{Common-neighbor Sentiment Paradox} & .592 & .571 \\
{Community Sentiment Paradox} & .590 & .569 \\
{Common-interest Sentiment Paradox} & .589 & .569 \\ 
{All Sentiment Paradoxes} & .617 & .600 \\
{All Sentiment Paradoxes + Friendship Paradox} & {.620} & {.600} \\
\bottomrule[1.5pt]
\end{tabular}}
\end{table}

\section{Related Work}
\label{sec::review}
Numerous studies have looked at network paradoxes, especially, friendship paradox. For example, friendship paradox has been observed in many online (e.g., Quora~\cite{iyer2018friendship} and Twitter~\cite{hodas2013friendship}) and offline networks~\cite{pires2017friendship}. Kooti et al.~\cite{kooti2014network} have observed and proved that friendship paradox must exist, based on the mean value, in social networks as node degrees always follow heavy-tail distributions.  
A recent study shows friendship paradox can help identify popular users by connecting it with friendship strength among users~\cite{bagrow2017friends}. 
Nettasinghe and Krishnamurthy utilize friendship paradox to design randomized polling methods for social networks~\cite{nettasinghe2018your}.

In addition to friendship paradox, recent literature has focused on the explorations of other network paradoxes such as user activity paradox, happiness paradox ~\cite{bollen2017happiness}, and scientific collaboration paradox indicating that researchers always have fewer coauthors, citations, publications, and lower h-index than their collaborators~\cite{fotouhi2014generalized,benevenuto2016h,eom2014generalized}. The development of these non-friendship paradoxes, however, is in an early stage, whose potential interpretations for their existence and applications have rarely investigated.

\section{Conclusion}
\label{sec:conclusions}

This work is motivated by the limitation of current sentiment analysis studies that have not considered interacting users in social networks, and by the phenomenon that people often consider their friends to be more positive than themselves, often attributed to human cognition biases in psychology. 
We present five sentiment paradoxes at the triad-, community- and network-level, all empirically and mathematically validated in undirected (i.e., with friendships) and directed (i.e., with follower and followee relationships) networks. Through studying the relations between the sentiment paradox and various characteristics of networks and users, we observe that (i) sentiment distributions determine the expected (non-) existence of sentiment paradoxes; (ii) node degrees (i.e., the number of social connections of users) is positively correlated to user sentiments; and (iii) there is no clear pattern between user sentiments and user activity. 
These connections (though not causal) can be responsible for the existence and magnitude of sentiment paradoxes in social networks, which cannot be solely attributed to human cognition bias as they generally exist in social networks. Additionally, we firstly demonstrate the application of our findings in predicting user sentiment prediction.
In the future, we will further analyze causal relationships between user's connections (degrees) and sentiments. Sentiment paradoxes in dynamic social networks as well as the ``like' and ``comments" networks will be part of our future studies.

\bibliographystyle{aaai}
\bibliography{references}

\section{Appendix}

\subsection{Post Distribution of Inactive Users}
In our experiments, we only retain users with ten or more posts to exclude occasionally active or inactive users. The post distribution of these excluded users is presented in Figure~\ref{fig:postDist}. The distribution indicates that a substantial number of users being not considered in our study has posted nothing.

\subsection{Illustration of User Post}
When posting blogs on LiveJournal, users can explicitly report their sentiments by selecting a mood. An illustration can be seen in Figure \ref{fig:post}, where the mood is \textit{Chipper}.

\subsection{Sentiment Polarity Identification of Moods}
There are 132 moods available on LiveJournal. The sentiment polarity (positive, neutral, or negative) of these moods is determined as shown in Table \ref{tab:mood}.

\begin{figure}[ht]
    \centering
    \includegraphics[width=0.4\textwidth]{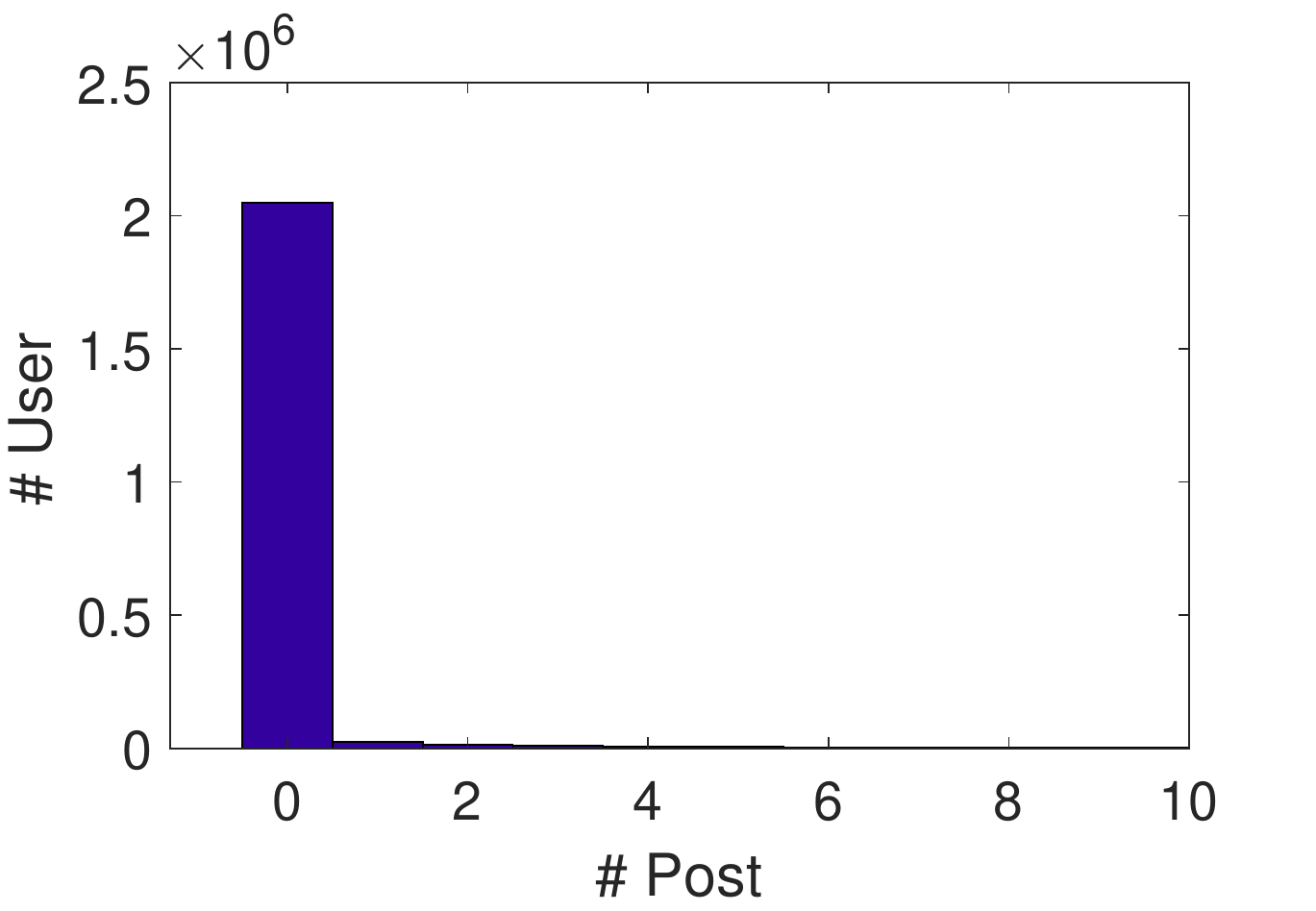}
    \caption{Post Distribution of Inactive Users} \label{fig:postDist}
\end{figure}

\begin{figure}[ht]
    \centering
    \includegraphics[width=0.45\textwidth]{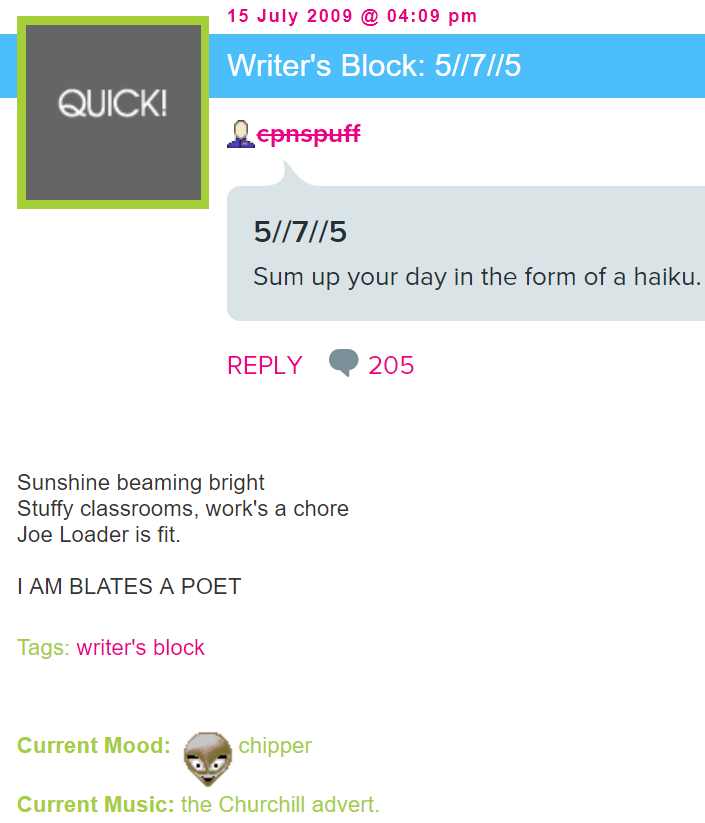}
    \caption{An Illustrated User Post with Mood \textit{Chipper}} \label{fig:post}
\end{figure}

\begin{table}[ht]
\small
\centering
\caption{Moods and Their Sentiment Polarity} \label{tab:mood}
\begin{tabular}{ll}
\toprule[1pt]
 & \multicolumn{1}{c}{\textbf{Mood}} \\ \hline
\multirow{7}{*}{\rotatebox{90}{\textbf{Positive}}} & amused; accomplished; artistic; bouncy; calm; cheerful; \\
 & content; creative; complacent; determined; excited; \\
 & ecstatic; energetic; full; good; giggly; grateful; happy; \\
 & hopeful; high; impressed; jubilant; loved; peaceful; \\
 & productive; pleased; rejuvenated; sympathetic; satisfied;\\
 & thankful; thoughtful; working; \\ \hline
\multirow{6}{*}{\rotatebox{90}{\textbf{Neutral}}} & awake; blah; blank; busy; chipper; contemplative; ditzy; \\
 & dorky; drained; drunk; flirty; geeky; groggy; horny; hot; \\
 & hyper; indescribable; intimidated; mellow; nerdy; okay; \\
 & optimistic; recumbent; refreshed; relaxed; rushed;  \\
 & shocked; sleepy; surprised; \\ \hline
\multirow{14}{*}{\rotatebox{90}{\textbf{Negative}}} & aggravated; angry; annoyed; anxious; apathetic; bitchy; \\
 & bored; cold; confused; cranky; crappy; crazy; crushed; \\
 & curious; cynical; depressed; devious; dirty; disappointed; \\
 & discontent; distressed; embarrassed; enthralled; envious; \\
 & exanimate; enraged; exhausted; frustrated; giddy; \\
 & gloomy; grumpy; guilty; hungry; indifferent; infuriated; \\
 & irate; irritated; jealous; lazy; lethargic; listless; lonely; \\
 & melancholy; mischievous; moody; morose; naughty; \\
 & nauseated; nervous; nostalgic; numb; pessimistic; pissed \\
 & off; pensive; predatory; quixotic; rejected; relieved;  \\
 & restless; sad; scared; sick; silly; sore; stressed; thirsty;  \\
 &  tired; touched; uncomfortable; weird; worried; \\
\bottomrule[1pt]
\end{tabular}
\end{table}

\end{document}